\documentclass[onecolumn,11pt,draftcls]{IEEEtran}

\usepackage{amsmath,stackrel}
\usepackage{amssymb}
\usepackage{amsthm}
\usepackage{arydshln}
\usepackage{subfig}
\usepackage{float}
\usepackage{graphicx}
\usepackage{dsfont}
\usepackage{xcolor}

\newtheorem{theorem}{Theorem}
\newtheorem{lemma}[theorem]{Lemma}
\newtheorem{definition}[theorem]{Definition}
\newtheorem{example}[theorem]{Example}

\newtheorem{remark}[theorem]{Remark}

\newtheorem{assumption}[theorem]{Assumption}

\newtheorem{corollary}[theorem]{Corollary}

\newcommand{{\Prob}}{\mathbb P}

\newcommand{\epi}{\mathrm{epi}}
\newcommand{\co}{\overline{\mathrm{co}}}

\newcommand{\mypar}[1]{\vspace{0.03in}\noindent{\bf #1.}}
\newcommand\mathitem{\item\leavevmode\vspace*{-\dimexpr\baselineskip+\abovedisplayskip\relax}}

\begin{document}

\title{Inaccuracy rates for distributed inference over random networks with applications to social learning}

\author{Dragana Bajovi\'c,~\IEEEmembership{Member,~IEEE} 
\thanks{D. Bajovi\'c is with the Department of Power, Electronics and Communications Engineering,  Faculty of Technical Sciences, University of Novi Sad. Email: dbajovic@uns.ac.rs.}

\thanks{Part of this work was done while the author was with the Department of Electrical and Computer Engineering, Carnegie Mellon University, Pittsburgh 15213, PA, USA~\cite{BajovicThesis13}.}
\thanks{This work is partially supported by the European Union’s Horizon 2020 Research and Innovation program under grant agreement No 957337. The paper reflects only the view of the author and the Commission is not responsible for any use that may be made of the information it contains.}
}

\maketitle
\begin{abstract}
This paper studies probabilistic rates of convergence for consensus+innovations type of algorithms in random, generic networks. For each node, we find a lower and also a family of upper bounds on the large deviations rate function, thus enabling the computation of the exponential convergence rates for the events of interest on the iterates. Relevant applications include error exponents in distributed hypothesis testing, rates of convergence of beliefs in social learning, and inaccuracy rates in distributed estimation. The bounds on the rate function have a very particular form at each node: they are constructed as the convex envelope between the rate function of the hypothetical fusion center and the rate function corresponding to a certain topological mode of the node's presence. We further show tightness of the discovered bounds for several cases, such as pendant nodes and regular networks, thus establishing the first proof of the large deviations principle for consensus+innovations and social learning in random networks.    
\end{abstract}

\section{Introduction}
\label{sec-Intro}

\IEEEPARstart{T}{he} theory of large deviations is the most prominent tool for studying \emph{rare events} that occur with stochastic processes, offering a principled approach for estimating probabilities of such events. A typical setup concerns a sequence of probability measures induced by the studied process and parameterized by one of the process parameters (e.g., time, population size, learning rate etc.), with the goal of computing, or characterizing, the respective decay rate, for any given event (region) of interest.  The practical value of such rates is in estimating the probability of a rare event of interest as an exponentially decaying function of the concerned process parameter, while neglecting the terms with slower than exponential dependence. The rates of rare events can additionally provide a ground for comparison of two statistical procedures, as originally proposed in the seminal work by Chernoff~\cite{Chernoff52}, and can therefore serve as a useful design criterion~\cite{Anandkumar07},~\cite{Bajovic11},~\cite{Tay15},~\cite{Ping22}. This is of special interest in the cases when other performance metrics are intractable for optimization, such as probabilities of error with hypothesis testing.   

In addition to the rate computation, large deviations analysis often reveals the most likely way through which the event of interest takes place, providing additional important insights that can guide system design. Most notable applications of large deviations theory are in statistics~\cite{Bucklew90}, communications and queuing theory~\cite{SchwartzWeiss95}, statistical mechanics~\cite{Touchette2009LDStatMechs}, and information theory~\cite{Cover91}.

For example, in statistical estimation, an event of interest is the event that the estimator does not belong to a predefined close neighborhood of the parameter being estimated~\cite{Arcones06LDM-estimators}.  The decay rates of probabilities of such events are known in the estimation theory as \emph{inaccuracy rates} and can, e.g., guide the decision on how many samples are needed for the estimator to reach the desired accuracy, with high probability~\cite{Bahadur60}. 
To make the exposition concrete, let $X_t \in \mathbb R^d$, $t=1,2,...$, be a sequence of estimators of a parameter $\theta \in \mathbb R^d$. Assuming that $X_t$ converges to $\theta$, an event of interest has the form $\{\|X_t-\theta\|\geq \epsilon\}$, where $\|\cdot\|$ denotes the $l_2$ norm (other vector norms can also be used). An equivalent way to represent this event is $\{X_t \in C_{\epsilon}\}$, where $C_{\epsilon}$ is the complement of the $l_2$ ball of diameter $\epsilon$ centered at $\theta$, $C_\epsilon = {B}^{\mathrm{c}}_{\theta}(\epsilon)$. Provided that $X_t$ converges to $\theta$, the probabilities of these events typically vanish exponentially fast with $t$.  Large deviations analysis then aims at discovering the corresponding rate of decay, i.e., the inaccuracy rate $\mathbf I(C_\epsilon)$:
\begin{equation}
\label{eq-rate-objective}
\mathbb P\left(X_t\in C_{\epsilon}\right) = e^{-t \mathbf I(C_\epsilon)+o(t)},
\end{equation}
where $o(t)$ denotes a function growing slower than linear with $t$. The inaccuracy rate $\mathbf I(C_\epsilon)$ has a very particular structure: it is given through the so called \emph{rate function} $I: \mathbb R^d \mapsto \mathbb R$ by
\begin{equation}\label{eq-set-fcn-via-rate-fcn}
\mathbf I(C_\epsilon) = \inf_{x\in C_{\epsilon}} I(x).
\end{equation}
The rate function $I$ is itself defined through the statistics of the inference sequence $X_t$. It should be noted that, in contrast with the set function $\mathbf I$, the rate function $I$ does not depend on the inaccuracy region, i.e., when $C_\epsilon$ varies, only the domain of minimization on the right-hand side of~\eqref{eq-set-fcn-via-rate-fcn} varies, while the rate function remains fixed. Also, this relation holds for an arbitrary set $C_{\epsilon}$ (e.g., not necessarily a ball complement). Hence, once the rate function is identified, the associated inaccuracy rate is readily computable through~\eqref{eq-set-fcn-via-rate-fcn} for a new given region of interest, without the need to redo the large deviations analysis each time, i.e., for each new region. Large deviations rate for estimation were first studied by Bahadur in~\cite{Bahadur60}.

Another well-known application of large deviations analysis is hypothesis testing~\cite{Chernoff52}, where the sequence $X_t$ is typically a decision statistics, e.g., obtained by summing up the log-likelihoods of the collected measurements up to the current time $t$, $X_t=1/t \sum_{s=1}^t \log \frac{f_1(Y_s)}{f_0(Y_s)}$; $f_0$ and $f_1$ here are the marginal distributions of the measurements $Y_s$ under the two hypotheses $H_0$ and $H_1$, respectively. If the acceptance threshold for $H_1$ at time $t$ is $\gamma_t$, then rare events of interest are $\{X_t < \gamma_t\}$, when $H_1$ is true (i.e., when $Y_s$ follow the distribution $f_1$) -- resulting in missed detection, and $\{X_t \geq \gamma_t\}$, when $H_0$ is true (when $Y_s$ follow the distribution $f_0$) -- causing a false alarm. When $C_{\epsilon}$ in~\eqref{eq-rate-objective} is replaced by the preceding two events, the resulting large deviations rates $\mathbf I (C_\epsilon)$ are then the well-known \emph{error exponents} that provide decay rates of the corresponding error probabilities.  

In this paper we are concerned with large deviations rates of \emph{distributed} statistical inference, where observations originate at different locations or different entities. Relevant works include algorithms such as consensus+innovations~\cite{GaussianDD},~\cite{Non-Gaussian-DD},~\cite{DDNoisy12},~\cite{DI-Directed-Networks16}, diffusion~\cite{MBMS16InfTheory},~\cite{MBMS16Refined},~\cite{MaranoSayed19OneBit}, and non-Bayesian or social learning~\cite{Jadb2012NonBayesian},~\cite{Shahrampour16DDFiniteTime},~\cite{Lalitha18SLandDHT},~\cite{Mitra21}. The common setup of the above works consists of networked nodes, each holding a local inference vector (parameter estimates, decision variables, beliefs) that is being updated over time. The updates are based on incorporating local, private signals that each agent observes over time, and then exchanging with immediate neighbors and averaging the received information through the well-known DeGroot averaging~\cite{DeGroot74} (also known as consensus). 

Asymptotic performance of distributed detection was studied in~\cite{GaussianDD}, for Gaussian observations,~\cite{Non-Gaussian-DD}, for generic observations, and in~\cite{DDNoisy12}, for networks with noisy communication links. In each of the named works, a randomly switching network topology is assumed and conditions for asymptotic equivalence of an arbitrary network node and a fusion center (with access to all observations) are studied. Reference~\cite{DI-Directed-Networks16} considers directed networks, both static and randomly varying,  and studies the rate function for the vector of states, deriving the exact rate function for the case of static networks, and providing bounds on the exponential rates for randomly switching networks. The rate function for static networks is given as the weighted combination of the local rate functions, with weights being equal to the eigenvector centralities (i.e., the left Perron vector of the consensus matrix). Reference~\cite{MBMS16InfTheory} studies distributed detection for static and symmetric networks and constant step size. For the limiting distribution of the local states, it proves the large deviations principle when the step size parameter decreases and shows that the rate function is equivalent to the centralized detector. These results are refined and extended in~\cite{MBMS16Refined} by studying non-exponential terms and directed (static) networks. Reference~\cite{MaranoSayed19OneBit} further considers distributed detection with 1-bit messages, while recent reference~\cite{Ping22} addresses optimal aggregation strategies for social learning.  

References~\cite{Jadb2012NonBayesian},~\cite{Shahrampour16DDFiniteTime},~\cite{Lalitha18SLandDHT},~\cite{Mitra21} study distributed $M$-ary hypothesis testing, where local updates are formed by applying Bayesian update on the vector of prior beliefs, based on the newly acquired local measurements. Assuming static, directed network, in~\cite{Jadb2012NonBayesian} and~\cite{Shahrampour16DDFiniteTime}, beliefs across immediate neighborhoods are merged through arithmetic average~\cite{Jadb2012NonBayesian}, while~\cite{Lalitha18SLandDHT} adopts geometric average (or, equivalently, arithmetic average on the log-beliefs). A different merging rule is proposed and analyzed in~\cite{Mitra21}, where instead of averaging, beliefs are updated by computing the minimum across the neighbors beliefs and the nodes' locally generated beliefs, showing improvement in the learning rate. Large deviations of the beliefs are addressed in~\cite{Lalitha18SLandDHT}, where it was proven that the log-ratios of beliefs with respect to the belief in the true distribution, satisfy the large deviations principle, with the rate function being equal to the eigenvector-centralities convex combination of the nodes' local rate functions, similarly as in~\cite{DI-Directed-Networks16} and~\cite{MBMS16Refined}. Through the contraction principle,~\cite{Lalitha18SLandDHT} also shows that the (log)-beliefs themselves satisfy the large deviations principle.  

\begin{figure}[thpb]%
 \centering
 \subfloat[Static topology]{\includegraphics[trim =42mm 90mm 42mm 90mm, clip, width=0.6\textwidth]{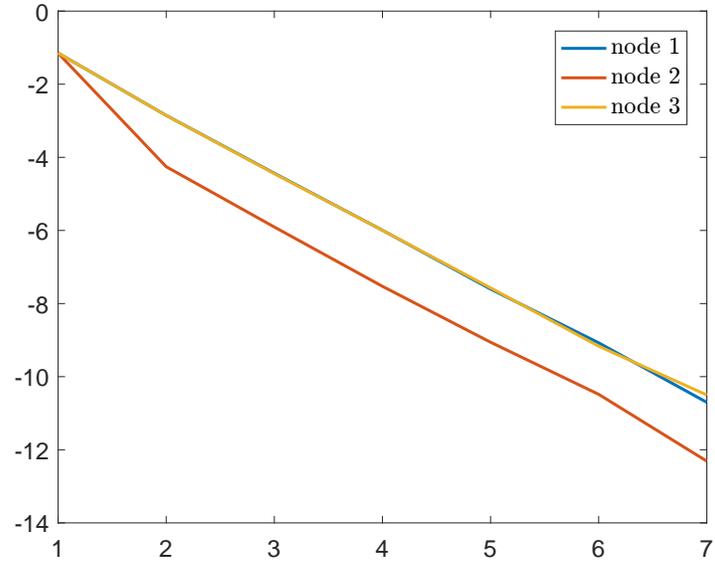}}\\%
 \subfloat[Randomly varying topology]{\includegraphics[trim =42mm 90mm 42mm 90mm, clip, width=0.6\textwidth]{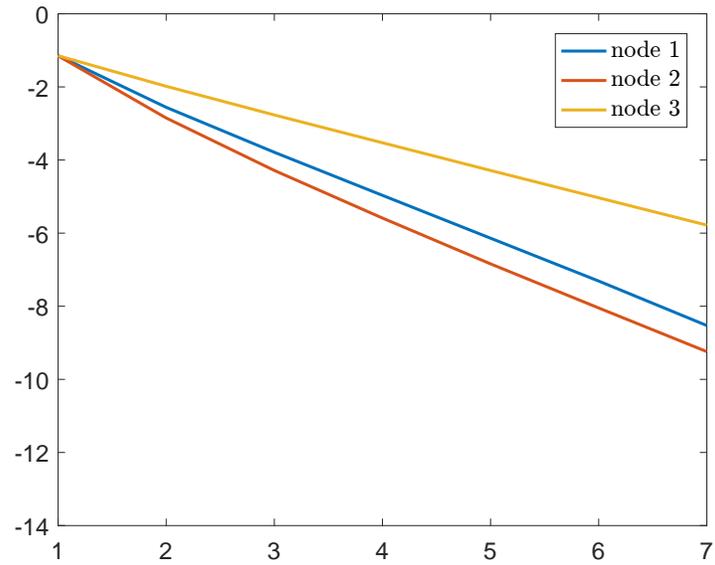}}
 \caption{Decay of the log-probabilities in~\eqref{eq-rate-objective} for a fixed set $C$ for static (top) and random (bottom) $3$ node chain network.}
 \label{Fig-Static-vs-random}%
\end{figure}

\mypar{Contributions} In contrast with the works in~\cite{MBMS16InfTheory}-\cite{Mitra21},  in this paper, we address computation of the rate function for distributed inference on \emph{random networks}. This model shift from static to random networks has fundamental implications on the large deviations performance. To explain this at an intuitive level: when the underlying network is random, consensus mixing of local inference vectors might be disabled for an arbitrary long period of time due to the lack of communications. In general, the topology can then break down into several connected components of the original network\footnote{Note that this is very different from time-varying networks that are typically modelled by the assumption of the so called bounded intercommunication interval, which guarantees that the union graph formed of all communication links occurring in this interval is connected, after a strictly finite time, e.g.,~\cite{Mitra21},~\cite{NedicSL-TV17}}. When in this regime, neither of the nodes can ``see'' the observations beyond the connected component they belong to, and hence the resulting rate function will be strictly lower than that of the full network\footnote{This is a consequence of the non-negativity of the rate function and the fact that it (roughly) scales linearly with the number of observation sources, as detailed in the paper.}. Figure~\ref{Fig-Static-vs-random} illustrates this effect with a toy example of a $3$-node chain where each node produces scalar observations of standard Gaussian distribution. 

In the top figure, we plot the logarithm of the probability in~\eqref{eq-rate-objective} for a ball complement inaccuracy set $C$, when the chain topology is fixed (static). In the bottom figure we plot the same probability, but when the two links of the chain graph alternate at random over time. We label the middle node as node 2, and we let the communication frequency between nodes 1 and 2 be higher (equal to $0.8$) than the one between higher than the one between nodes 2 and 3 (equal to $0.2$). It is clear from the figure that the static topology achieves much steeper decay, and, moreover, this decay is equal at each of the three nodes (and also equal to the decay of the hypothetical fusion center, cf. Section~\ref{sec-Main}, as predicted by the theory). In contrast, in the random case, the difference between the nodes' decays is evident: node 2 achieves the steepest decay, followed by node $1$, while node $3$ has the worst performance.

In this work, we are interested in understanding the rate function of each node in the network and analytically expressing its dependence on the system parameters. For each node, we find a lower and a family of upper bounds on the rate function. This is achieved by carrying out node-specific large deviations analyses. We show that the two bounds match in several cases, such as for pendant nodes and also for nodes in a regular network. The family of upper bounds is indexed by different induced components of the given node, and each function in this family has the form of the convex envelope between the rate function of the full network and the rate function of the respective component, lifted up by the probability of the event that induces the component. The lower bound is given as the convex envelope between the rate function of the full network and the node's local rate function lifted up by the large deviations rate of consensus, whose existence was shown in~\cite{Rate-of-consensus13}. With respect to references~\cite{GaussianDD},~\cite{Non-Gaussian-DD},~\cite{DDNoisy12}, there are several important novelties. First, we extend the results of~\cite{Non-Gaussian-DD} to the case of vector observations and vector inference state. Second, while~\cite{GaussianDD}-\cite{DDNoisy12} only provide a lower bound on the rate function, this work, as described in the above, finds also a family of upper bounds. This is achieved by carefully devising events that impact the rate function, and for which we develop novel large deviations techniques. The discovered upper bounds enable to establish, to the best of our knowledge, the first proof of the large deviations principle for nodes performing DeGroot-based distributed inference in randomly varying networks.  

As an application of particular interest to this study, we consider social learning, specifically the form with the geometric average update~\cite{Lalitha18SLandDHT}.  We show that, with appropriate transformation of the belief iterates -- namely, considering their log-ratios with respect to the belief in the true distribution, the algorithm studied in~\cite{Lalitha18SLandDHT}  exhibits full equivalence to the consensus+innovations algorithm that we analyze here. Building on this equivalence, we characterize the rate function of the beliefs in social learning and provide the first proof of the large deviations principle for social learning run over random networks. 

A closely related work to ours is~\cite{Parasnis20} that studies convergence properties of social learning over random networks. This reference shows that, almost surely, each node is able to correctly identify the true hypothesis. We similarly focus on the case of random networks, but we are additionally concerned with characterizing the \emph{rates} of probabilistic convergence  of the iterates in the sense of large deviations. Finally, we show that almost sure convergence of the beliefs follows from the obtained large deviations rates. 

From the technical perspective, this paper contributes with a novel set of techniques and approaches that could be of interest for further studies of social learning, and more generally, distributed inference in random networks.


\mypar{Notation} For arbitrary $d\in \mathbb N$ we denote by $0_d$ the $d$-dimensional vector of all zeros; by $1_d$ the $d$-dimensional vector of all ones; by $e_i$ the $i$-th canonical vector of $\mathbb R^d$ (that has value one on the $i$-th entry and the remaining entries are zero); by $I_d$ the $d$-dimensional identity matrix; by $J_d$ the $d\times d$ matrix whose all entries equal to $1/d$. For a matrix $A$, we let $[A]_{ij}$ and $A_{ij}$ denote its $i,j$ entry and  for a vector $a\in \mathbb R^d$, we denote its $i$-th entry by $a_i$, $i,j=1,...,d$. For the set of indices $C\subseteq \{1,2,...,N\}$, we let $[A]_C$ (or $A_C$) denote the submatrix of $A$ that corresponds to indices in $C$.
For a function $f:\mathbb R^d\mapsto \mathbb R$, we denote its domain by $\mathcal D_f=\left\{ x\in \mathbb R^d: -\infty <f(x)<+\infty \right\}$; for a set $D \subseteq \mathbb R$, $f^{-1}(D)$ is defined as $f^{-1}(D) = \{x\in \mathbb  R^d: f(x)\in D\}$. $\log$ denotes the natural logarithm. 
For $N\in \mathbb N$, we denote by $\Delta_{N-1}$ the probability simplex in $\mathbb R^N$ and by $\alpha$ the generic element of this set: $\Delta_{N-1}=\left\{\alpha \in \mathbb R^N: \alpha_i \geq 0, \sum_{i=1}^N \alpha_i=1\right\}$. We let $\lambda_{\max}$ and $\lambda_2$, respectively, denote the maximal and the second largest (in modulus) eigenvalue of a square matrix; $\|\cdot\|$ denotes the spectral norm. For a matrix $S\in \mathbb R^{N\times N}$, we let $\mathcal R(S)$ denote the range of $S$, $\mathcal R(S)=\left\{Sx: x\in \mathbb R^N\right\}$. An open Euclidean ball in $\mathbb R^d$ of radius $\rho$ and centered at $x$ is denoted by $B_{x}(\rho)$; the closure, the interior, and the complement of an arbitrary set $D\subseteq \mathbb R^d$ are respectively denoted by $\overline D$, $D^{\mathrm{o}}$, and $D^{\mathrm{c}}$; $\mathcal B(\mathbb R^d)$ denotes the Borel sigma algebra on $\mathbb R^d$; $\mathbb P$ and $\mathbb E$ denote the probability and the expectation operator; $\mathcal N(m,S)$ denotes Gaussian distribution with mean vector $m$ and covariance matrix $S$. For a given graph $H$, $E(H)$ denotes the set of edges of $H$.

\mypar{Paper organization} Section~\ref{sec-model} describes the system model and the algorithm and Section~\ref{sec-LD-metric} introduces the large deviations metric and defines the relevant large deviations quantities. Section~\ref{sec-Main} states the main result of the paper, important corollaries and provides illustration examples. Section~\ref{sec-SL} provides applications of the results to social learning. Proofs of the main result are given in Section~\ref{section-proofs}. 
Section~\ref{sec-Concl} concludes the paper.

\section{System model}
\label{sec-model}
This section explains the system model and the consensus+innovations distributed inference algorithm accompanied by different application examples. Section~\ref{subsec-SL} details the connection to social learning, while Section~\ref{subsec-graph-preliminaries} provides certain preliminaries. 

\mypar{Communication model} We consider a network of $N$ identical agents connected by an arbitrary communication topology. The topology is represented by an undirected graph $\overline G=(V, \overline E)$, where $V=\left\{1,2,...,N\right\}$ is the set of agents, and $\overline E\subseteq {V \choose 2}$ is the set of possible communication links between agents. We assume that during operation of the network each link $\{i,j\}\in\overline E$ may fail, and that correlations between failures of different links are possible. Realization (i.e., a snapshot) of the communication topology at time slot $t$ is denoted by $G_t=(V,E_t)$, for $t=1,2,\ldots,$ where $E_t$ is the set of links that are online at time $t$; note that $E_t\subseteq \overline E$. For an agent $i$, we let $O_{i,t}$ denote the set of neighbors of $i$ at time $t$, $O_{i,t}=\left\{j\in V: \{i,j\}\in E_t\right\}$. %

\mypar{Consensus based distributed estimation} At each time $t$, each sensor $i$ acquires a $d$-dimensional vector of measurements $Z_{i,t}\in \mathbb R^d$. We assume that the measurements $Z_{i,t}$ are independent and identically distributed across sensors and over time. The goal of each sensor is to estimate the state of nature $\theta$, which is the expected value of sensor observations $Z_{i,t}$, $\theta=\mathbb E\left[Z_{i,t}\right]$. To achieve this, an agent $i$ holds a local estimate, called also the state, $X_{i,t}$ and iteratively updates it over time slots $t$. At each slot $t$, agent $i$ performs two steps: 1) the innovation step; and 2) the consensus step. In the innovation step, $i$ acquires $Z_{i,t}$ and incorporates it into the current state $X_{i,t-1}$, by computing the following convex combination, forming an intermediate state:
\begin{equation}
\label{alg-1}
\widehat X_{i,t}= \frac{t-1}{t}X_{i,t-1}+\frac{1}{t} Z_{i,t}.
\end{equation}
It then subsequently transmits $\widehat X_{i,t}$ to (possibly, a subset of) its neighbors in $\overline G$, and, at the same time, receives the intermediate states $\widehat X_{j,t}$, $j\in O_{i,t}$,  from its current neighbors. In the second, consensus, step, agent $i$ computes the convex combination (DeGroot averaging) between its own and the neighbors estimates:
\begin{equation}
\label{alg-2}
X_{i,t}= \sum_{j\in O_{i,t}\cup\{i\}} W_{ij,t} \widehat X_{j,t},
\end{equation}
where $W_{ij,t}$ is the weight that agent $i$ at time $t$ assigns to the estimate of agent $j$. For neat exposition, the weights of all nodes are collected in an $N$ by $N$ matrix $W_t$, such that the $i,j$ entry of $W_t$ equals $W_{ij,t}$, when $j\in O_{i,t}\cup\{i\}$, and equals zero otherwise. Thus, $W_t$ respects the sparsity pattern of $G_t$: if $\{i,j\}\notin E_t$, then $[W_t]_{ij}=[W_t]_{ji}=0$. Also, since the weights at each node form a convex combination, matrix $W_t$ is stochastic. In addition, we assume that, at any time $t$, for any $i,j$, the weights are symmetric at each link, i.e., $W_{ij,t}=W_{ji,t}$, implying that $W_t$ is symmetric.

Denoting by $\Phi(t,s)=W_t\cdots W_s$ for $1\leq s\leq t$, algorithm~\eqref{alg-1}-\eqref{alg-2} can be written as:
\begin{equation}
\label{alg-compact}
X_{i,t}=\frac{1}{t}\,\sum_{s=1}^t \sum_{j=1}^N [\Phi(t,s)]_{i,j}\, Z_{j,s}.
\end{equation}
We analyse algorithm~\eqref{alg-1}-\eqref{alg-2} under the following assumptions on the matrices~$W_t$ and observations~$Z_{i,t}$.
\begin{assumption}[Network and observations random model]
\label{ass-network-and-observation-model}
\leavevmode
\begin{enumerate}
\item Observations $Z_{i,t}$, $i=1,\ldots,N$, $t=1,2,\ldots$ are independent, identically distributed (i.i.d.) across nodes and over time;
\item \label{ass-W-t}
The sequence of matrices $W_t$, $t=1,2,\ldots$ is i.i.d. and for each $t$, every realization of $W_t$ is stochastic, symmetric and has positive diagonals;
\item \label{ass-connected}
$\lambda_2\left( \mathbb E\left[W_t\right]\right)<1,$ or, equivalently, the induced graph $\overline G$ of $\mathbb E\left[W_t\right]$ is connected.
\item Weight matrices $W_t$ are independent from the nodes' observations $Z_{i,s}$ for all $i$, $s$, $t$.
\end{enumerate}
\end{assumption}
We now present different application examples of algorithm~\eqref{alg-1}-\eqref{alg-2}.

\begin{example} [Estimating the distribution of opinions by social sampling]
\label{ex-SOC}
Consider the scenario where a group of $N$ agents wishes to discover the distribution of opinions (e.g., about an event or phenomenon) across a certain, large population. To achieve this, agents continuously poll the population and register responses of individuals. We assume that the respondents' opinions are quantized to $d$ preset opinion summaries:~$\left\{r_1,...,r_d\right\}$. We let $\mathcal R_{i,t}$ denote the opinion (summary) of the person that agent $i$ interviewed at time $t$. Also, let $p_l$ be the probability that the response of a person chosen uniformly at random is $r_l$.  Consider now algorithm~\eqref{alg-1}-\eqref{alg-2} and define the innovation vector $Z_{i,t}$ to be the vector of opinion indicators, $Z_{i,t}=\left( 1_{\{\mathcal R_{i,t}=r_1\}},...,1_{\{\mathcal R_{i,t}=r_d\}}\right)^\top$; again, let the $W_t$'s be arbitrary stochastic matrices. Then, the states of all agents converge to the true opinion distribution, $\left(p_1,\ldots,p_d\right)$, as we show in Section~\ref{sec-Main}, i.e., algorithm~\eqref{alg-1}-\eqref{alg-2} is able to correctly identify the distribution of opinions across a given population, while the rates of this convergence will prove to be highly dependent on the frequency of agents' interactions and interaction patterns.
\end{example}

\begin{example}[Distributed event detection]
\label{ex-D-DET}
Suppose that a wireless sensor network is deployed in a certain area to detect in which of the two possible states the environment is. This problem can be modeled as a binary hypothesis testing problem, where under the state of nature (hypothesis) $\mathbf H_1$, the sensors measurements follow the distribution $f_1$, and similarly for $f_0$, where $f_1$ and $f_0$ are assumed known. We let $Y_{i,t}$ denote the measurement of sensor $i$ at time $t$. We assume that $Y_{i,t}$'s are independent both over time and across different sensors. This hypothesis testing problem can be solved by algorithm~\eqref{alg-1}-\eqref{alg-2} as follows. For each $i$ and $t$, define the innovation $Z_{i,t}$ as the log-likelihood ratio of the node $i$'s measurement at time $t$: $Z_{i,t}= \log \frac{f_1\left(Y_{i,t}\right)} {f_0\left(Y_{i,t}\right)}$. Then, any sensor in the system can, at any given time, make a decision simply by comparing its state $X_{i,t}$ against a prescribed threshold $\gamma$:
\begin{equation}
X_{i,t}\stackrel[\mathcal H_0]{\mathcal H_1}{ \gtreqless}\gamma.
\end{equation}
For further details on distributed detection application, see also~\cite{Non-Gaussian-DD}.
\end{example}

A generalization of the preceding example to $M$-ary hypothesis testing and an application to social learning is given in the next subsection.

\subsection{Social learning}
\label{subsec-SL}
The idea of social learning is for a group of people to distinguish between $M$ different hypotheses, potentially indistinguishable by any given individual,  through local Bayesian updates and collaborative information exchange. Each node $i$ over time draws observations $Y_{i,t}$ from (the true) distribution $f_{i,M}$ (hypothesis $\mathbf H_M$); the remaining $M-1$ candidate distributions that compete  at node $i$ in hypothesis testing are $f_{i,m}$ (hypothesis $\mathbf H_m$), $m=1,...,M-1$. It is assumed that, conditioned on the true hypothesis $\mathbf H_M$, observations at each node are independent over time, and they are also independent from the observations that are generated at any different node. 

We consider here the algorithm for social learning proposed in~\cite{Lalitha18SLandDHT}.  Each node $i$ maintains over time two sets of values (vectors), $q_{i,t}\in \mathbb R^M$ and $b_{i,t}\in \mathbb R^M$, called, respectively, \emph{private} and \emph{public belief} vectors, quantifying node $i$'s beliefs in each of the $M$ hypotheses. The $m$-th entry of $q_{i,t}$, denoted by $q_{i,t}^m\in \mathbb R$, corresponds to the private belief of node $i$ in the $m$-th hypothesis; similarly, the $m$-th entry of $b_{i,t}$, denoted by $b_{i,t}^m\in \mathbb R$, corresponds to the public belief of node $i$ in the $m$-th hypothesis. The values of both public and private belief vectors are between $0$ and $1$: the closer an entry of a belief vector is to $1$ ($0$), the stronger (weaker) is the confidence of the respective node that the corresponding hypothesis is true;  e.g., if for some $m$, $b_{i,t}^m$ equals $1$, this means that node $i$ is fully confident that hypothesis $\mathbf H_m$ is true. 

The algorithm starts at each node with initial private beliefs $q_{i,t}^m>0$, $m=1,...,M-1$. Upon receiving new local observation $Y_{i,t}$, each node $i$ updates its $m$-th public belief as follows: 
\begin{equation}
\label{alg-SL-1}
b_{i,t}^m= \frac{f_{i,m}(Y_{i,t}) q_{i,t-1}^{m}}{\sum_{l=1}^M f_{i,l}(Y_{i,t}) q_{i,t-1}^{l}},    \end{equation}
for each $m=1,...,M$. The node then sends its updated  public belief vector $b_{i,t} = (b_{i,t}^1,...,b_{i,t}^M)^\top$ to all of its neighbors $O_{i,t}$. Upon receiving the neighbors' (public) beliefs, the node updates its private beliefs as follows:
\begin{equation}
\label{alg-SL-2}
 q_{i,t}^{m} = \frac{e^{\sum_{j\in O_{i,t}} W_{ij,t} \log b_{j,t}^{m}}}{ \sum_{l=1}^M e^{ \sum_{j\in O_{i,t}} W_{ij,t} b_{j,t}^{l} }},   
\end{equation}
for each $m=1,...,M$. 

It is easy to verify that both $q_{i,t}$ and $b_{i,t}$ represent valid probability vectors, i.e., $q_{i,t},\,b_{i,t}\in \Delta_{M-1}$.  

\mypar{Connection with algorithm~\eqref{alg-1}-\eqref{alg-2}} Consider the update for the private belief $q_{i,t}^m$ in~\eqref{alg-SL-2}. Computing the log-ratios of $q_{i,t}^{m}$ with ${q_{i,t}^{M}}$ (belief in the true hypothesis $\mathbf H_M$),  the recursion in~\eqref{alg-SL-2} transforms into:
\begin{equation}
\label{alg-SL-2-transformed}
\log \frac {q_{i,t}^{m}}{{q_{i,t}^{M}}} = \sum_{j\in O_{i,t}} W_{ij,t} \log \frac {b_{j,t}^{m}}{{b_{j,t}^{M}}}.   
\end{equation}
Similarly, it is easy to see that the log-ratios of the public beliefs $b_{j,t}^{m}$ with $b_{j,t}^{M}$ can be expressed as:
\begin{equation}
\label{alg-SL-1-transformed}
\log \frac {b_{i,t}^{m}}{{b_{i,t}^{M}}} =  \log \frac {q_{i,t-1}^{m}}{{q_{i,t-1}^{M}}} +  \log \frac{f_{i,m}(Y_{i,t})}{f_{i,M}(Y_{i,t})}.  
\end{equation}
Dividing both sides in~\eqref{alg-SL-2-transformed} and~\eqref{alg-SL-1-transformed} by $t$, we recognize the form in~\eqref{alg-1}-\eqref{alg-2}. Further, denoting, for each $m=1,...,M-1$, 
\begin{align}
\label{eq-Z-i-t-SL}
L_{i,t}^m &= \log \frac{f_{i,m}(Y_{i,t})}{f_{i,M}(Y_{i,t}}\\
\widehat {X}_{i,t}^m & = \frac{1}{t} \log \frac {q_{i,t}^{m}}{{q_{i,t}^{1}}}\\
\label{eq-X-i-t-SL}
X_{i,t}^m & = \frac{1}{t} \log \frac {b_{i,t}^{m}}{{b_{i,t}^{1}}}
\end{align}
 and stacking the per-hypothesis quantities in vector form: $L_{i,t} = \left( L_{i,t}^1,...,L_{i,t}^{M-1}\right)\in \mathbb R^{M-1}$, and $\widehat {X}_{i,t} = \left( \widehat {X}_{i,t}^1,..., \widehat {X}_{i,t}^{M-1}\right)\in \mathbb R^{M-1}$, and $X_{i,t} = \left( X_{i,t}^1,...,X_{i,t}^{M-1}\right)\in \mathbb R^{M-1}$, the exact form in~\eqref{alg-1}-\eqref{alg-2} is obtained, where the innovation vectors $Z_{i,t}$ that algorithm~\eqref{alg-1}-\eqref{alg-2} is fed with are the log-likelihood ratio vectors $L_{i,t}$; note also that, in this application instance, $d=M-1$. Thus,  the generic algorithmic form~\eqref{alg-1}-\eqref{alg-2}  subsumes also the social learning algorithm from~\eqref{alg-SL-1}-\eqref{alg-SL-2} through the described variable transformation. Section~\ref{sec-SL} shows how results of this paper can be used to characterize convergence of beliefs and large deviations rates of social learning, specifically for the case when the weights $W_{ij,t}$ (neighborhoods $O_{i,t}$) in~\eqref{alg-SL-2} are random.  

\subsection{Probabilistic rate of consensus $\mathcal J$}
\label{subsec-graph-preliminaries}
We next define certain concepts and quantities pertinent to the underlying graph process that are needed for later analyses.

\mypar{Components in union graphs} Since the sequence of matrices $W_t$ is i.i.d., the sequence $G_t$ of their underlying topologies is i.i.d. as well. We let $\mathcal G$ denote the set of all topologies on $V$ that have non-zero probability of occurrence at a given time $t$, i.e., $\mathcal G=\left\{ (V,E): \mathbb P\left( G_t=(V,E)\right)>0\right\}$. For convenience, for any undirected, simple graph $H$ on the set of vertices $V$ we denote $p_H=\mathbb P\left(G_t=H\right)$. Thus, for any $H\in \mathcal G$, $p_H>0$. It will also be of interest to consider different subsets of the set of feasible graphs $\mathcal G$. For a collection of undirected simple graphs $\mathcal H$ on $V$ we let $\Gamma_{\mathcal H}=(V,E_{\mathcal H})$ denote the corresponding union graph, that is, $\Gamma_{\mathcal H}$ is the graph with the set of vertices $V$ and whose edge set $E_{\mathcal H}$ is the union of edge sets of all the graphs in $\mathcal H$, $E_{\mathcal H}= \cup_{H \in \mathcal H} E(H)$. We let $p_{\mathcal H}$ denote the probability that $G_t$ belongs to $\mathcal H$, \[p_{\mathcal H}=\sum_{H\in \mathcal H} p_H.\]
We also introduce  -- what we refer to as -- the component of a node in $\mathcal H$.
\begin{definition}[Node component in union graph]
\label{def-union-component}
Let $\mathcal H$ be a given collection of undirected simple graphs on $V$ and let $C_1,...,C_L$ be the components of the union graph $\Gamma(\mathcal H)$. Then, the component of node $i$ in $\mathcal H$, denoted by $C_{i,\mathcal H}$, is the component of $\Gamma(\mathcal H)$ that contains $i$: i.e., if $i\in C_l$, then $C_{i,\mathcal H}=C_l$.
\end{definition}

\mypar{Probabilistic rate of consensus $\mathcal J$} We recall here the rate of consensus, associated with a sequence of random stochastic symmetric matrices, introduced in~\cite{GaussianDD} and subsequently analyzed in~\cite{Rate-of-consensus13}. In~\cite{GaussianDD} and~\cite{Non-Gaussian-DD} we showed that the quantity $\mathcal J$ below, termed the rate of consensus\footnote{The rate of consensus $\mathcal J$ (in~\eqref{def-rate-of-consensus}) is defined slightly differently than the corresponding quantity from~\cite{GaussianDD} and~\cite{Non-Gaussian-DD}. In~\cite{GaussianDD} and~\cite{Non-Gaussian-DD}, in the event $\|W_t\cdots W_1-J_N\|>1/t$, the probability of which we wish to compute, there is a constant $\varepsilon\in (0,1]$ in the place of $1/t$ . However, as we show in~\cite{Rate-of-consensus13}, the two rate quantities coincide when the weight matrices are i.i.d., which is the case that we consider here.}, captures well how the weight matrices $W_t$ affect performance of the estimates $X_{i,t}$ when one is concerned with large deviations metrics:
\begin{equation}
\label{def-rate-of-consensus}
\mathcal J:= - \limsup_{t\rightarrow +\infty}\, \frac{1}{t}\,\log \mathbb P\left( \left\| W_t\cdots W_1-J\right\| > \frac{1}{t} \right).
\end{equation}
%
Rate of consensus $\mathcal J$ is computed exactly in~\cite{Rate-of-consensus13}.
\begin{theorem}[\cite{Rate-of-consensus13}]
\label{theorem-compute-rate-of-consensus}
Let Assumption~\ref{ass-network-and-observation-model}, part~\ref{ass-W-t} hold. Then the $\limsup$ in~\eqref{def-rate-of-consensus} is in fact a limit and the rate of consensus $\mathcal J$ is found by
 \[\mathcal J = |\log p_{\mathcal H^\star}|,\]
 where $p_{\max}$ is the probability of the most likely collection of feasible graphs whose union graph is disconnected,
\begin{equation}
\mathcal H^\star=\arg \max_{\mathcal H\subseteq \mathcal G: \,\Gamma_{\mathcal H}\;\mathrm{disc.}} p_{\mathcal H}.
\end{equation}
\end{theorem}
In the next example we consider an important special case when links in $\overline G$ fail independently at random.

\begin{example}[Random topologies with i.i.d. link failures]
\label{ex-Iid-failures}
Consider the random model for $W_t$ defined by Assumption~\ref{ass-network-and-observation-model}.\ref{ass-W-t} where each link in $\overline G$ fails independently from other links with probability $1-p$. Applying Theorem~\ref{theorem-compute-rate-of-consensus}, it can be shown that 
\begin{equation}
    \mathcal J = \mathrm{min\,cut}\, (\overline G) |\log(1-p)|,
\end{equation} where $\mathrm{min\,cut}\,(\overline G)$ is the minimum edge cut of the graph $\overline G$; for example, if $\overline G$ is a chain, then $\mathrm{min\,cut}\,(\overline G)=1$. The details of this derivation can be found in~\cite{Rate-of-consensus13}. 
\end{example}

For finite time analyses, of relevance is the following variant of~\eqref{def-rate-of-consensus}: for any $\epsilon>0$, there exists a positive constant $K_{\epsilon}$ such that for all $t$,
\begin{equation}
\label{eq-rate-of-consensus-epsilon}
\mathbb P\left( \left\| W_t\cdots W_s-J_N\right\| > \frac{1}{t} \right) \leq K_{\epsilon}e^{- (t-s)\,(\mathcal J-\epsilon)}.
\end{equation}

\section{Problem formulation: The metric of large deviations}
\label{sec-LD-metric}
Section~\ref{sec-model} illustrate uses of algorithm~\eqref{alg-1}-\eqref{alg-2} for several applications: multi-agent polling with cooperation, in Example~\ref{ex-SOC}, fully distributed hypothesis testing, in Example~\ref{ex-D-DET}, and social learning, in Section~\ref{subsec-SL}. We now introduce the rates of large deviations that we adopt as performance metric for applications of algorithm~\eqref{alg-1}-\eqref{alg-2}.

\mypar{Rate function~$I$ and the large deviations principle}
\begin{definition}[Rate function $I$~\cite{DemboZeitouni93}]
\label{def-Rate-function}
Function $I:\mathbb R^d\mapsto [0,+\infty]$ is called a \emph{rate function} if it is lower semicontinuous, or,
equivalently, if its level sets are closed. If, in addition, the level sets of $I$ are compact (i.e., closed and bounded), then $I$ is called a good rate function.
\end{definition}

\begin{definition}[The large deviations principle~\cite{DemboZeitouni93}]
\label{def-LDP}
Suppose that $I:\mathbb R^d\mapsto [0,+\infty]$ is lower semicontinuous. A sequence of measures $\mu_t$ on $\left(\mathbb R^d,\mathcal B\left(\mathbb R^d\right)\right)$, $t\geq 1$, is said to satisfy the large deviations principle (LDP) with rate function~$I$ if, for any measurable set $D\subseteq \mathbb R^d$, the following two conditions hold:
\begin{enumerate}
\item\label{eqn-LDP-UB}
$\displaystyle \limsup_{t\rightarrow +\infty}\,\frac{1}{t}\,\log\mu_t(D)\leq -\,\inf_{x\in \overline D} I(x);$
\item\label{eqn-LDP-LB}
$\displaystyle \liminf_{t\rightarrow +\infty}\,\frac{1}{t}\,\log\mu_t(D)\geq -\,\inf_{x\in D^{\mathrm{o}}} I(x).$
\end{enumerate}
\end{definition}

Differently than with the case of static topologies, when topologies and/or weight matrices $W_t$ are random, finding the rate function of an arbitrary node performing distributed inference is a very difficult problem~\cite{Non-Gaussian-DD,Soummya-LDRiccati14}. (In fact, even the existence of the LDP is not known a priori.) Our approach is to find functions $\overline I_i$ and $\underline I_i:\mathbb R^d\mapsto \mathbb R$, such that, for any measurable set $D$:
\begin{align}
\label{eq-rate-UB}
\limsup_{t\rightarrow +\infty}\,\frac{1}{t}\,\log\,\mathbb P\left(X_{i,t} \in D  \right) &\leq - \inf_{x\in \overline D} \underline I_i(x),\\
\label{eq-rate-LB}
\liminf_{t\rightarrow +\infty}\,\frac{1}{t}\,\log\,\mathbb P\left(X_{i,t} \in D  \right) &\geq - \inf_{x\in D^{\mathrm{o}}} \overline I_i(x).
\end{align}
At a high level, this is analytically achieved by carefully constructing events the probabilities of which upper and lower bound the probability of the event of interest in~\eqref{eq-rate-UB} and~\eqref{eq-rate-LB}. We remark that functions $\underline I_i$ and $\overline I_i$ that we seek should satisfy~\eqref{eq-rate-UB} and~\eqref{eq-rate-LB} for any given set $D$, i.e., similarly as with the rate function $I_i$, to find bounds on the exponential rates for a given rare event $\{X_{i,t}\in D\}$, it suffices to perform minimizations of $\underline I_i$ and $\overline I_i$ over $D$. This property is very important, as once $\underline I_i$ and $\overline I_i$ are discovered, any inaccuracy rate can be easily estimated without the need to do any (further) large deviations analyses.  

As we show in Appendix~\ref{app:rate-function-bounds}, if for some node $i$ the LDP holds and~\eqref{eq-rate-UB} and~\eqref{eq-rate-LB} are satisfied for any $D$, then
\begin{equation}
\label{eq-rate-function-bounds}
\underline I_i(x)\leq I_i(x)\leq \overline I_i(x),\;x\in \mathbb R^d,
\end{equation}
i.e., the graph of the LDP rate function $I_i$ lies between the graphs of $\overline I_i$ and $\underline I_i$.

\mypar{Log-moment generating function of observations~$Z_{i,t}$ and its conjugate} We proceed standardly by introducing the log-moment generating function of the observation vectors~$Z_{i,t}$, which we denote by~$\Lambda$. The log-moment generating function $\Lambda:\,{\mathbb R}^d \rightarrow \mathbb R \cup \{+\infty\}$ corresponding to $Z_{i,t}$ is defined by:
\begin{equation}
\label{def-Lmgf}
\Lambda(\lambda)=\log \mathbb E\left[ e^{\lambda^\top Z_{i,t}}\right],\:\:\mathrm{for\:\:}\lambda \in \mathbb R^d.
\end{equation}
We make the assumption that $\Lambda$ is finite at all points.%
\begin{assumption}
\label{ass-finite-at-all-points}
$\mathcal D_{\Lambda}= {\mathbb R}^d$, i.e., $\Lambda(\lambda)<+\infty$ for all $\lambda \in \mathbb R^d$.
\end{assumption}

Besides the log-moment generating function $\Lambda$, the second key object in large deviations analysis is the Fenchel-Legendre transform, or the conjugate, of~$\Lambda$, defined by
\begin{equation}
\label{def-Conjugate}
I(x)=\sup_{\lambda\in \mathbb R^d} x^\top \lambda - \Lambda(\lambda),\:\:\mathrm{for\:\:}x\in \mathbb R^d.
\end{equation}
Log-moment generating function and its conjugate enjoy many nice properties, such as convexity and differentiability in the interior of the function's domain~\cite{DemboZeitouni93},~\cite{Hollander}. We list the properties that are relevant for the current analysis in the next lemma. Recall that $\theta = \mathbb E[Z_{i,t}]$.

\begin{lemma}[Properties of $\Lambda$ and $I$]
\label{lemma-properties}
\leavevmode
\begin{enumerate}
    \item \label{part-Lambda-cvxity} $\Lambda$ is convex and differentiable on $\mathbb R^d$;
    \item \label{part-Lambda-at-0} $\Lambda(0)=0$ and $\nabla \Lambda (0)=\theta$;
    \item \label{part-I-cvxity} $I$ is strictly convex; 
    \item \label{part-Lambda-I-gradient-connection} if $x=\nabla \Lambda(\lambda)$ for some $\lambda\in \mathbb R^d$, then $I(x)= \lambda^\top x - \Lambda(\lambda)$;
    \item \label{part-I-at-m} $I (x)\geq 0$ with equality if and only if $x=\theta$. 
\end{enumerate}
\end{lemma}
 
 Proofs of~\ref{part-Lambda-cvxity}-\ref{part-I-at-m} (with a weaker form of the claim in part~\ref{part-I-cvxity} -- with strict convexity replaced by convexity, and with non-negativity only in part~\ref{part-I-at-m}) can be found in~\cite{DemboZeitouni93}. The proof of strict convexity of $I$ under Assumption~\ref{ass-finite-at-all-points} can be found in~\cite{Vysotsky21}.  We briefly comment on properties~\ref{part-Lambda-at-0} and~\ref{part-I-at-m}, to give some (mathematical) intuition as to why these properties hold, where we note that of particular, practical relevance is~\ref{part-I-at-m}.   Plugging in $\lambda=0$ in the defining equation of $\Lambda$,~\eqref{def-Lmgf}, it is easy to see that $\Lambda(0)=0$. Similarly, it can be shown that, for any $\lambda$, $\nabla \Lambda(\lambda) = \mathbb E[Z_{i,t} e^{\lambda^\top Z_{i,t}}]/\mathbb E [e^{\lambda^\top Z_{i,t}}]$. Evaluating at $\lambda=0$, the property  $\nabla \Lambda (0)=\theta$ follows. Property~\ref{part-I-at-m} has a very intuitive meaning: the rate function is non-negative and also equals zero at the mean value. To see why the latter holds, it suffices to invoke properties from part~\ref{part-Lambda-at-0} in~\ref{part-Lambda-I-gradient-connection}; note also that, since $I$ is non-negative, $\theta$ is a minimizer of $I$. The if and only if part then follows from strict convexity of $I$, which implies uniqueness of its minimizer $\theta$. We will show practical implications of this property when considering large deviations rate of the sequence $X_{i,t}$. 

The following result, proven in~\cite{DI-Directed-Networks16}, gives fundamental large deviations upper and lower bound for the inference sequence $X_{i,t}$. The result holds for arbitrary stochastic weight matrices $W_t$ and, in particular, for directed topologies as well. This result will be invoked when proving tightness and optimality of our rate function bounds for certain classes of networks, in Section~\ref{subsec-Interpretations-and-Corollaries}.
\begin{lemma}[Fundamental distributed inference bounds]
\label{lemma-fundamental-bounds}
Consider algorithm~\eqref{alg-1}-\eqref{alg-2} under Assumptions~\ref{ass-network-and-observation-model} and~\ref{ass-finite-at-all-points}. Then~\eqref{eq-rate-UB} and~\eqref{eq-rate-LB} hold with $\overline I_i=NI$ and $\underline I_i=I$, for all $i$.
\end{lemma}

\mypar{Closed convex hull of a function}
We recall the definitions of the epigraph and closed convex hull of a function.
\begin{definition}[Epigraph and closed convex hull of a function,~\cite{Urruty}]
\label{def-cvx-hull}
Let $f : \mathbb R^d \mapsto \mathbb R\cup \{+ \infty\}$ be a given function.
\begin{enumerate}
\item The epigraph of $f$, denoted by $\mathrm{epi} f$, is defined by
\begin{equation}
\mathrm{epi} f = \left\{ (x,r): r\geq f(x), x\in \mathbb R^d\right\}.
\end{equation}
\item Consider the closed convex hull $\overline{ \mathrm {co}}\, \mathrm{epi}\, f$\footnote{The convex hull of a set $A$, where $A$ is a subset of some Euclidean space, is defined as the set of all convex combinations of points in $A$~\cite{Urruty}.} of the epigraph of $f$. The closed convex hull of $f$, denoted by $\overline{\mathrm {co}}f$, is defined by:
\begin{equation}
\overline{\mathrm {co}} f(x):= \inf\{r: (x,r) \in \overline{\mathrm {co}}\, \mathrm{epi}\, f\}.
\end{equation}
\end{enumerate}
\end{definition}
Hence, for a given function $f$, epigraph of $f$ is the area above the graph of $f$. Closed convex hull of $f$ is then constructed from $\mathrm{epi} f$ by first finding the closed convex hull of the epigraph, $\overline{\mathrm {co}}\, \mathrm{epi}\, f$. Then, $\overline{\mathrm {co}} f$ is defined as the function the epigraph of which matches $\overline{\mathrm {co}}\, \mathrm{epi}\, f$. Intuitively,  $\overline{\mathrm {co}}f$ is the best convex and lower semi-continuous (closed) approximation of $f$, as its epigraph contains (besides $\mathrm{epi}\, f$) only those points that are needed for ``convexification'' and closure. Figure~\ref{fig-cvx-hull-Gaussian} further ahead gives an illustration of $\overline{\mathrm {co}}f$, while construction of $\overline{\mathrm {co}}f$ is explained in Section~\ref{subsec-a-closer-look-at-conjugate-functions-from-the-theorem}.

%

\section{Main result}
\label{sec-Main}
The main result of this section, Theorem~\ref{theorem-nice-tight-bounds}, finds functions $\underline I_i$ and $\overline I_i$ from~\eqref{eq-rate-UB} and~\eqref{eq-rate-LB}. These functions enable computation of bounds on the exponential decay rate of an arbitrary rare event and, in the case of the existence of the LDP, by~\eqref{eq-rate-function-bounds}, provide approximations to the rate function $I_i$. A number of important corollaries of Theorem~\ref{theorem-nice-tight-bounds} is then presented in Subsection~\ref{subsec-Interpretations-and-Corollaries}, including the large deviations principle for regular networks and for pendant nodes. 
Section~\ref{sec-SL} then studies application of the derived results to distributed hypothesis testing and social learning. 

\begin{theorem}
\label{theorem-nice-tight-bounds}
Consider distributed inference algorithm~\eqref{alg-1}-\eqref{alg-2} under Assumptions~\ref{ass-network-and-observation-model} and~\ref{ass-finite-at-all-points}. Then, for each node~$i$, for any measurable set $D$:
\begin{enumerate}
\mathitem \label{part-upper-bound}
\begin{align}
\label{eq-UB-Theorem}
& \limsup_{t\rightarrow +\infty}\,\frac{1}{t}\,\log \mathbb P\left(X_{i,t}\in D\right) \leq - \inf_{x\in \overline D} I^\star (x),
\end{align}
where $I^\star (x) = \overline{\mathrm {co}}\,\inf\left\{I(x)+\mathcal J, N I(x)\right\}$;
\item \label{part-lower-bound} for any collection $\mathcal H$ of graphs on $V$:
\begin{align}
\label{eq-LB-Theorem}
& \liminf_{t\rightarrow +\infty}\,\frac{1}{t}\,\log \mathbb P\left(X_{i,t}\in D\right) \geq - \inf_{x\in D^{\mathrm{o}}} I_{i,\mathcal H}(x),
\end{align}
where $I_{i,\mathcal H} (x) = \overline{\mathrm {co}}\,\inf \left\{ |C_{i,\mathcal H}| I(x)+ \left|\log p_{\mathcal H}\right|, N I(x)\right\}$.
\end{enumerate}
\end{theorem}
In words, Theorem~\ref{theorem-nice-tight-bounds} asserts that, for a fixed set $D$, for any node $i$, the probabilities $\mathbb P\left(X_{i,t}\in D\right)$ decay exponentially fast over iterations $t$ and it also finds bounds on the rate of this decay. We now make a couple of additional remarks and such that aim at gaining further insights and intuition about this result and the relevant quantities.
 
\begin{remark}
Consider an arbitrary disconnected collection $\mathcal H$. By the construction of  $C_{i,\mathcal H}$, for any node $i$, there holds $\{i\}\subseteq C_{i,\mathcal H}$ and, by non-negativity of $I$, it follows that $I \leq |C_{i,\mathcal H}| I$ (point-wise). Further, from Theorem~\ref{theorem-compute-rate-of-consensus} we know that $\mathcal J= |\log p_{\mathcal H^\star}|\leq |\log p_{\mathcal H}|$. Therefore, we have that for any disconnected collection $\mathcal H$, $I+\mathcal J \leq |C_{i,\mathcal H}| I + |\log p_{\mathcal H}|$. The latter obviously implies $I^\star \leq I_{i,\mathcal H}$, serving as a first feasibility check for~\eqref{eq-rate-function-bounds} (and also~\eqref{eq-rate-UB} and~\eqref{eq-rate-LB}).  
\end{remark} 

Comparing the upper bound from Theorem~\ref{theorem-nice-tight-bounds} with~\eqref{eq-rate-UB}, we see that~\eqref{eq-rate-UB} is satisfied for
\begin{equation}
\underline I_i\equiv I^\star, \mbox{\;for\; all\;} i\in V.
\end{equation}
That is, we have a uniform (lower) bound $I^\star$ on each of the nodes' rate functions $I_i$, $i\in V$.

With respect to the lower bound from Theorem~\ref{theorem-nice-tight-bounds}, there is in fact a whole family of functions $\overline I_i$, one per each collection of graphs $\mathcal H$, that validate~\eqref{eq-rate-LB}. To find the best bound for a given $D$, we might optimize the right hand side of~\eqref{eq-LB-Theorem} over all collections $\mathcal H$. This, however, might be computationally infeasible. Instead, we can focus only on those collections $\mathcal P\subseteq \mathcal G$ that have a certain property, e.g., $\mathcal P=\left\{ \mathcal H: |C_{i,\mathcal H}| =n\right\}$, for some $n$, $1\leq n\leq N$. Then, $\overline I_i$ from~\eqref{eq-rate-LB} can be found by finding $\mathcal H\in \mathcal P$ that yields uniformly lowest (i.e., closest to $I_i$) $I_{i,\mathcal H}$:
\begin{equation}
\overline I_i =  \inf_{\mathcal H\in \mathcal P}  I_{i,\mathcal H}, \mbox{\;for\;} i\in V.
\end{equation}
The following corollary follows directly from~\eqref{eq-rate-function-bounds} and the definition of LDP.
\begin{corollary}
\label{corollary-sandwich-I-i}
\begin{enumerate}
\item If, for a given $i$, the sequence $X_{i,t}$, $t=1,2,...$ satisfies the LDP with rate function $I_i$, then, for any collection of graphs $\mathcal H$,
\begin{equation}
\label{eq-sandwich-I-i}
I^\star \leq I_i \leq I_{i,\mathcal H}.
\end{equation}
\item If, for a given $i$, for some $\mathcal P$ (possibly, a single element set $\mathcal P=\left\{\mathcal H\right\}$), $I^\star \equiv \inf_{\mathcal H\in \mathcal P}  I_{i,\mathcal H}$, then the sequence $X_{i,t}$, $t=1,2,...$ satisfies the LDP with rate function $I_i=I^\star \equiv \inf_{\mathcal H\in \mathcal P}  I_{i,\mathcal H}$.
\end{enumerate}
\end{corollary}
In the next remark, through simple convex analyses, we make a connection between Corollary~\ref{corollary-sandwich-I-i} (Theorem~\ref{theorem-nice-tight-bounds}) and Lemma~\ref{lemma-fundamental-bounds}, completing the established bounds in~\eqref{eq-sandwich-I-i} with the general bounds from Lemma~\ref{lemma-fundamental-bounds}, hence establishing a coherent view of the derived results. 

\begin{remark} [Recovery of fundamental bounds in Lemma~\ref{lemma-fundamental-bounds}]
\label{remark-I-NI-I_i-implied-bounds}
From the point-wise non-negativity of $I$ and non-negativity of $\mathcal J$, it is easy to see that $I\leq NI$ and $I \leq I+\mathcal J$. Thus, $\mathrm{epi} \inf\{NI, I+ \mathcal J\}\subseteq \mathrm{epi} I$. Since $I$ is closed and convex, $\overline{\mathrm {co}}\, \mathrm{epi} I = \mathrm{epi} I$, thus implying $\overline{\mathrm {co}}\,\mathrm{epi} \inf\{NI, I+ \mathcal J\}\subseteq \mathrm{epi} I$. The latter directly implies $I \leq I^\star$. Similarly, we have $NI \geq \inf\{NI, |C_{i,\mathcal H}|I+ |\log p_{\mathcal H}|\}$, where the latter holds for any disconnected collection $\mathcal H$. Thus $\mathrm{epi} NI \subseteq \mathrm{epi} \inf\{NI, |C_{i,\mathcal H}|I+ |\log p_{\mathcal H}|\}$, which in turn implies $\overline{\mathrm {co}}\, \mathrm{epi} NI \subseteq \overline{\mathrm {co}}\, \mathrm{epi} \inf\{NI, |C_{i,\mathcal H}|I+ |\log p_{\mathcal H}|\}$. Since $NI$ is convex and closed (the properties inherited from $I$), $\overline{\mathrm {co}}\, \mathrm{epi} NI = \mathrm{epi} NI$, and therefore $\mathrm {epi} NI = \overline{\mathrm {co}}\, \mathrm{epi} NI \subseteq \overline{\mathrm {co}}\, \mathrm{epi} \inf\{NI, |C_{i,\mathcal H}|I+ |\log p_{\mathcal H}|\}$. The latter implies $NI \geq I_{i,\mathcal H}$. Combining with~\eqref{eq-sandwich-I-i} establishes: 
\begin{equation}
\label{eq-sandwich-I-i-fundamental}
I \leq I^\star \leq I_i \leq I_{i,\mathcal H}\leq NI.
\end{equation}
The above chain of inequalities is a capture of the so far established bounds in the literature on the large deviations rate function for consensus+innovations distributed inference iterates on random networks. 

As a byproduct, we note in passing that~\eqref{eq-sandwich-I-i-fundamental} verifies  Lemma~\ref{lemma-fundamental-bounds} for the special case of stochastic symmetric weight matrices.
\end{remark}

\begin{remark} [Zero rate at $\theta$]
\label{remark-zero-rate-at-theta}
 Since $I$ is non-negative, both $NI$ and $|C_{i,\mathcal H}|I(\theta) + |\log p_H|$ are also non-negative, implying $I_{i,\mathcal H}\geq 0$. Further, from Lemma~\ref{lemma-properties}, we have $I(\theta)=0,$ and noting now that $NI(\theta)=0 < |C_{i,\mathcal H}|I(\theta) + |\log p_H|$, it follows that $I_{i,\mathcal H}(\theta)=0$.  It can be similarly shown that $I^\star(\theta)=0$. From the preceding properties it follows that for any set $C$ containing the mean value $\theta$
 \begin{equation}
\inf_{x \in C} I^\star (x) = \inf_{x \in C} I_{i,\mathcal H}(x)=0.       
 \end{equation}
 It follows that $\mathbf I_i(C)=0$, i.e., the inaccuracy rate for any $C$ containing $\theta$ equals zero. This means that probabilities of events that $X_{i,t}$ belong to $C$ do not exhibit an exponential decay -- specifically, for any norm ball centered at $\theta$, and of an arbitrary radius $\rho>0$, $B_{\theta}(\rho)>0$, there holds
\begin{equation}
\label{eq-zero-rate}
\lim_{t\rightarrow +\infty}\,\frac{1}{t}\,\log \mathbb P\left(X_{i,t}\in B_\theta(\rho)\right) =0.
\end{equation}
Observing the form of the algorithm, eq.~\eqref{alg-1}-\eqref{alg-2}, where innovations $Z_{i,t}$ -- the mean vector of which is $\theta$, are incorporated and mixed via weighted averaging (both over time and across nodes), it is intuitive to expect that $X_{i,t}$ will converge to $\theta$ (consider the ideal averaging case -- $W_t = J_d$, for which $X_{i,t} = \frac{1}{t} \sum_{s=1}^t \sum_{j=1}^N \frac{1}{N} Z_{j,s}$, which converges to $\theta$ by the law of large numbers). Hence, the zero decay in~\eqref{eq-zero-rate} is intuitive, i.e., the probabilities that $X_{i,t}$ belongs to a neighborhood of $\theta$ should not vanish with $t$.   
\end{remark}

We use the result of Theorem~\ref{theorem-nice-tight-bounds}, together with the uniqueness of the minimizer of $I$, property~\ref{part-I-at-m} from Lemma~\ref{lemma-properties}, to establish a sort of a converse to~\eqref{eq-zero-rate} - i.e., whenever we seek the inaccuracy rate $\mathbf I_i(C)$ for a set $C$ not containing $\theta$, this rate will be strictly positive. Practical relevance of this (technical) property is given in Theorem~\ref{theorem-almost-sure-convergence} below, where almost sure convergence of $X_{i,t}$ to $\theta$ is formally established.

\begin{remark} [Strictly non-zero rate at $x \neq \theta$]
Consider an arbitrary point $x\neq \theta$. From Lemma~\ref{lemma-properties}, part~\ref{part-I-at-m} we know that $I(x)>0$ for any $x\neq \theta$. 

Consider now an arbitrary set $C$ such that $\theta \notin C$. By strict convexity of $I$ and uniqueness of the minimizer of $I$, it follows that $I$ is coercive~\cite{Rockafellar2015}. Pick an arbitrary point $x_0 \in C$ and let $\alpha=I(x_0)$. Define $S_{\alpha}= \{x\in \mathbb R^d: I(x)\leq \alpha\}$, i.e., $S_{\alpha}$ is the $\alpha$-level set of $I$. By coercivity of $I$, it follows that $S_{\alpha}$ is compact. We now note   
\begin{equation}
\label{eq-infimum-grt-0-1}
\inf_{x\in C} I(x) = \inf_{x \in C\cap S_{\alpha}} I(x)=:a.
\end{equation}
Compactness of $S_{\alpha}$ implies compactness of $C\cap S_{\alpha}$ and since $I$ is continuous and strictly greater than $0$, it follows by the Weierstrass theorem that the infimum of $I$ over $C$ is strictly greater than zero, $a=\inf_{x \in C\cap S_{\alpha}} I(x)>0$. Finally, By the fact that $I^\star \geq I$ (the left-hand side inequality in~\eqref{eq-sandwich-I-i-fundamental}), we in turn obtain:
\begin{equation}
\label{eq-infimum-grt-0-2}
\inf_{x\in C} I^\star(x) \geq \inf_{x\in C} I(x) = a >0.
\end{equation}
Therefore, for any set $C$ such that $\theta \notin C$, we have
\begin{equation}
\label{eq-nonzero-rate-a}
\limsup_{t\rightarrow +\infty}\,\frac{1}{t}\,\log \mathbb P\left(X_{i,t}\in C\right) \leq - a<0,
\end{equation}
where the constant $a$ bounding the exponential decay rate depends on the chosen set $C$.   
\end{remark}

With preceding considerations at hand, almost sure convergence of nodes' iterates $X_{i,t}$ follows by standard arguments. 
\begin{theorem} [Almost sure convergence of $X_{i,t}$]
\label{theorem-almost-sure-convergence}
Consider distributed inference algorithm~\eqref{alg-1}-\eqref{alg-2} under Assumptions~\ref{ass-network-and-observation-model} and~\ref{ass-finite-at-all-points}. Then, for each node~$i$, the state vectors $X_{i,t}$ converge almost surely to $\theta=E[Z_{i,t}].$
\end{theorem}

\begin{proof}
Fix node $i\in V$. Pick an arbitrary $\epsilon>0$ and consider $C = \mathbb B^{\mathrm{c}}_{\theta}(\epsilon)$. We start by noting that inequality in~\eqref{eq-nonzero-rate-a} implies existence of a finite $t_0=t_0(C)$ such that, for all $t\geq t_0$, $\mathbb P(X_{i,t} \in C) \leq e^{-t \frac{a}{2}}$.  Then, for all $t\geq t_0$, we have
\begin{equation}
\mathbb P\left(\|X_{i,t}-\theta\| \geq \epsilon \right) \leq e^{-t\frac{a}{2}}.    
\end{equation}
Thus,
\begin{equation}
\mathbb P\left(\|X_{i,t}-\theta\| > \epsilon, \mathrm{\,i.o.} \right) \leq \sum_{t=1}^\infty e^{-t\frac{a}{2}}<\infty,    
\end{equation}
where the last inequality follows from strict positivity of $a$. Applying the Borel-Cantelli lemma~\cite{Karr}, the claim of the theorem follows.
\end{proof}

\subsection{A closer look at functions $I^\star$ and $I_{i,\mathcal H}$}
\label{subsec-a-closer-look-at-conjugate-functions-from-the-theorem}

This subsection finds closed form expressions for the functions $I^\star$ and $I_{i,\mathcal H}$ for the case when $Z_{i,t}$ is a Gaussian vector, and provides a graphical interpretation of the obtained result.

\begin{lemma}
\label{lemma-Gaussian-design}
Let $Z_{i,t}$ be Gaussian with mean vector $m$ and covariance matrix $S$. Then
\begin{equation}
\label{eqn-characterization-of-I-star}
I^\star(x)=\left\{
 \begin{array}{ll}
NI(x),& x \in \mathcal R_1^\star \\
N \sqrt{2 c_1} \,H(x)- N c_1, & x \in \mathcal R_2^\star \\
I(x)+\mathcal J, & x \in \mathcal R_3^\star
\end{array}\right.,
\end{equation}
where $\mathcal R_1^\star = \left\{x: N I(x)\leq  c_1\right\}$, $\mathcal R_2^\star = \left\{x: c_1<I(x)\leq  N c_1\right\}$,  and $\mathcal R_3^\star = \left\{x: I(x)>  N c_1\right\}$, $I(x)=\frac{1}{2}(x-m)^\top S^{-1}(x-m)$, $H(x)= \sqrt{(x-m)^\top S^{-1}(x-m)}$, and $c_1=\frac{\mathcal J}{N(N-1)}$. Also, for any fixed collection of graphs $\mathcal H$
\begin{equation}
\label{eqn-characterization-of-I-i-mathcal-H}
I_{i,\mathcal H}(x)=\left\{
 \begin{array}{ll}
NI(x),& x \in \mathcal R_1^{i,\mathcal H} \\
N \sqrt{2 c_2}\, H(x)- N c_2 , & x \in \mathcal R_2^{i,\mathcal H}\\
|C_{i,\mathcal H}|\,I(x)+|\log p_{\mathcal H}|, & x\in \mathcal R_3^{i,\mathcal H}
\end{array}\right.,
\end{equation}
where $\mathcal R_1^{i,\mathcal H}= \left\{x: \frac{N}{|C_{i,\mathcal H}|} I(x)\leq  c_2\right\}$, $\mathcal R_2^{i,\mathcal H}= \left\{x: c_2<I(x)\leq  \frac{N}{|C_{i,\mathcal H}|} c_2\right\}$, $\mathcal R_3^{i,\mathcal H}= \left\{x: I(x)>  \frac{N}{|C_{i,\mathcal H}|} c_2 \right\}$,  and $c_2=\frac{|C_{i,\mathcal H}| |\log p_{\mathcal H}|}{N(N-|C_{i,\mathcal H}|)}$.
\end{lemma}
Proof of Lemma~\ref{lemma-Gaussian-design} is given in Appendix~\ref{app:compute-I-star}.

\mypar{Three regions of $I^\star$} We provide a graphical illustration for ${I^\star}$ in Figure~\ref{fig-cvx-hull-Gaussian}. We consider an instance of algorithm~\eqref{alg-1}-\eqref{alg-2} running on a $N=3$-node chain, with i.i.d. link failures of probability $(1-p)=e^{-5}$, and where the observations $Z_{i,t}$ are standard Gaussian (zero mean and variance equal to one). For standard Gaussian, $I(x)=\frac{1}{2}x^2$, and we obtain from Example~\ref{ex-Iid-failures} that the rate of consensus equals $\mathcal J=|\log (1-p)|=5$.  The more curved blue dotted line plots the function $NI(x)= \frac{1}{2}N x^2$, the less curved blue dotted line plots the function $I(x)+\mathcal J= \frac{1}{2}x^2+5$, and the solid red line plots ${I^\star}$. Observing the figure and the corresponding formula~\eqref{eqn-characterization-of-I-star}, we see that  ${I^\star}$ is defined by three regions. In the region around the zero mean, $\mathcal R_{1}^\star$, ${I^\star}$ matches the optimal rate function $NI$. On the other hand, in the outer region, $\mathcal R_3^\star$, where values of $x$ are sufficiently large, ${I^\star}$ follows the slower growing function, $I+\mathcal J$. Finally, in the middle region, $\mathcal R_2^\star$, ${I^\star}$ is linear (more generally, when $d>1$, $I^\star$ will exhibit linear intervals over any direction that crosses the mean value). This linear part is the tangent line that touches both the epigraph of $NI(\cdot)$ and the epigraph of $I+\mathcal J$ and is responsible for the convexification of the point-wise infimum  $\inf\left\{I+\mathcal J, NI \right\}$. Function $I_{i,\mathcal H}$ has similar properties.
\begin{figure}[thpb]
  \centering
  \includegraphics[trim=5.5cm 4cm 4cm 4cm, clip=true, totalheight=7cm]{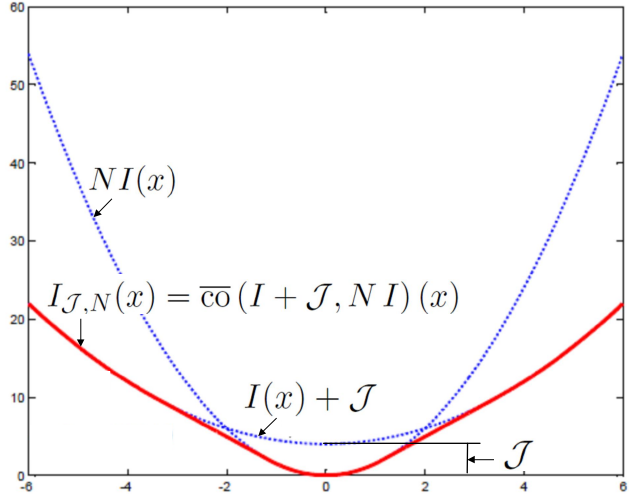}
  \caption{Illustration of ${I^\star}$ for a chain network of size $N=3$, with $\mathcal J=5$, and $Z_{i,t}\sim\mathcal N(0,1)$. The more curved blue dotted line plots $NI(x)= \frac{1}{2}N x^2$, the less curved blue dotted line plots $I(x)+\mathcal J= \frac{1}{2}x^2+\mathcal J$. The solid red line plots ${I^\star}=\mathrm{co} \left( NI, I+\mathcal J \right)$.}
\label{fig-cvx-hull-Gaussian}
\end{figure}

\subsection{Illustrations and LDP for special cases}
\label{subsec-Interpretations-and-Corollaries}

In this subsection, we use Theorem~\ref{theorem-nice-tight-bounds} to establish the LDP for certain classes of random models. As explained in the remarks after Theorem~\ref{theorem-nice-tight-bounds}, to prove the LDP at some node $i$, it is sufficient to show that $I^\star$ and $I_{i,\mathcal H}$ coincide for some collection $\mathcal H$.

The first corollary of Theorem~\ref{theorem-nice-tight-bounds} asserts that if every realization of the network topology is connected, then, for any node $i$, the sequence of states $X_{i,t}$ satisfies the LDP with rate function $N I$. In our recent work~\cite{DI-Directed-Networks16}, we prove that $NI$ is the best (highest) possible rate function for any distributed inference algorithms of the form~\eqref{alg-1}-\eqref{alg-2} with $N$ nodes. It is also the rate function of a hypothetical fusion node that has access to all the observations. Thus, when every instance of the network topology is connected, then each node in the network is, in the asymptotic sense, effectively acts as a fusion center. Corollary~\ref{corollary-every-topology-connected} was, for the special case of Gaussian observations, previously proved in~\cite{Allerton11}.
\begin{corollary}
\label{corollary-every-topology-connected}
Let, for each $t$, $G_t$ be connected. Then, for any $i\in V$, $X_{i,t}$ satisfies the large deviations principle with rate function $N I$.
\end{corollary}
\begin{proof}
By Theorem 2 from~\cite{DI-Directed-Networks16}, we know that, for any node $i$ and for any set $D$
\begin{equation}
\liminf_{t\rightarrow +\infty} \,\frac{1}{t}\log \mathbb P\left( X_{i,t}\in D \right) \geq - \inf_{x\in D^{\mathrm{o}}} N I(x).
\end{equation}
Comparing with the conditions for LDP in Definition~\ref{def-LDP}, we see that we only need to prove that $I^\star\equiv N I$. For the latter identity it suffices to show that $\mathcal J=+\infty$, because then $\inf\{ NI, I+\mathcal J\}\equiv NI$, and since $NI$ is closed and convex, we obtain $I^\star=\co (NI)= NI$. Suppose for the sake of contradiction that there exists a disconnected collection of graphs $\mathcal H$ such that $p_{\mathcal H}>0$. Then, there must be a graph $H\in \mathcal H$ such that both $H$ is disconnected and $p_{H}>0$. But this contradicts the assumption that every possible (i.e., non-zero probability) topology is connected. Thus, it must be that for every disconnected collection $p_{\mathcal H}=0$ implying $\mathcal J=+\infty$, and proving the claim.
\end{proof}
In particular, Corollary~\ref{corollary-every-topology-connected} implies that if the nodes' interactions are deterministic, i.e., $W_t\equiv A$, for some stochastic symmetric $A$, and $A$ is such that $\left|\lambda_2(A)\right|<1$, then, for each $i$, $X_{i,t}$ satisfy the LDP with the optimal rate function $N I$. This recovers the large deviations principle for deterministic networks, established in~\cite{DI-Directed-Networks16}, for the special case of symmetric networks (cf. Theorem~1 in~\cite{DI-Directed-Networks16}).

\mypar{LDP for critical nodes} Consider now a situation when there exists a node $i$ such that $\mathcal J= \left| \log p_{i,\mathrm{isol}}\right|$, where $p_{i,\mathrm{isol}}$ denotes the probability that $i$ operates in isolation due to network randomness, $p_{i,\mathrm{isol}}=\mathbb P\left( O_{i,t}=\emptyset\right)$. Comparing with Theorem~\ref{theorem-compute-rate-of-consensus}, this means that the most likely way to disconnect $\overline G$ is to isolate~$i$, i.e., 
\begin{equation}
p_{\max}=\sum_{H \in \mathcal H_{i,\mathrm{isol}}} p_H,
\end{equation}
where $\mathcal H_{i,\mathrm{isol}}= \left\{ H:p_{H}>0,\,\,C_{i,H}=\{i\}\right\}$. Since $C_{i,\mathcal H_{i,\mathrm{isol}}}=\{i\}$, we have $\left|C_{i,\mathcal H_{i,\mathrm{isol}}}\right|=1$. Consider now the lower bound in~\eqref{eq-LB-Theorem} for $\mathcal H=\mathcal H_{i,\mathrm{isol}}$. Noting that $\left|p_{\mathcal H_{i,\mathrm{isol}}}\right|= \mathcal J$, we see that the two functions $I^\star$ and $I_{i,\mathcal H_{i,\mathrm{isol}}}$ coincide, thus implying the LDP for node $i$. This is formally stated in the next corollary.
\begin{corollary}[LDP for critical nodes]
\label{corollary-aLDP-for-the-critical-node}
Suppose that for some $i$, $\mathcal J= \left| \log p_{i,\mathrm{isol}}\right|$. Then, the sequence of states $X_{i,t}$ satisfies the LDP with the rate function $\co\left\{ N I(x), I(x)+ \left|\log p_{i,\mathrm{isol}}\right| \right\}$.
\end{corollary}

In the next two corollaries we assume the random model from Assumption~\ref{ass-network-and-observation-model}.\ref{ass-W-t} where each link in the graph $\overline G$ fails independently with the same probability $1-p$, $p\in [0,1]$.

\begin{corollary}[LDP for pendant nodes]
\label{cor-pendant}
Suppose that the random model for $W_t$ is such that all links in $\overline E$ fail independently from each other with probability $1-p$. Then, for any node~$i$ whose degree in $\overline G$ is equal to one, its sequence of states $X_{i,t}$ satisfies the LDP with the rate function $\co\left\{ N I(x), I(x)+ \left|\log(1-p)\right| \right\}$.
\end{corollary}
\begin{proof}
Suppose that $i$ is a degree one node. By Corollary~\ref{corollary-aLDP-for-the-critical-node}, it suffices to show that $\mathcal J=|\log (1-p)|$. From Example~\ref{ex-Iid-failures}, we know that $\mathcal J$ equals $|\log (1-p)|$ times the minimum edge cut of $\overline G$. In this case, minimum edge cut equals one (and is achieved, for instance, when the edge adjacent to $i$ is removed from the network), which proves the result.
\end{proof}

\begin{corollary}[LDP for regular networks]
\label{cor-regular}
Suppose that $\overline G$ is a circulant network in which each node is connected to $d/2$ nodes on the left and $d/2$ nodes on the right, where $d\leq N-1$ is even. We assume that each link, independently of all other links, fails with probability $1-p$. Then, for any node $i$ its sequence of states $X_{i,t}$ satisfies the LDP with the rate function $\co\left\{ N I, I+ d \log \left|1-p\right| \right\}$.
\end{corollary}
\begin{proof}
Note that $p_{i,\mathrm{isol}}=(1-p)^d$ for any $i$. Hence, by Corollary~\ref{corollary-aLDP-for-the-critical-node}, it suffices to show that $\mathcal J=d |\log(1-p)|$. Observing that the minimum cut in this case equals $d$, the result follows.
\end{proof}

\section{Application to distributed hypothesis testing and social learning}
\label{sec-SL}

In this subsection we show how results from Section~\ref{sec-Main} can be used to characterize large deviations rates of distributed hypothesis testing and social learning that are run over random networks. We recall the algorithm and relevant quantities defined in Section~\ref{subsec-SL}. We assume that the measurement distributions corresponding to the same hypothesis are equal across all nodes, i.e., when hypothesis $\mathbf H_m$ is true, the measurements at all nodes are drawn from the same distribution $f_m$: $Y_{i,t}\sim f_{i,m}\equiv f_m$, for all $i$.  

Following the identified role of the vector of log-likelihood ratios $L_{i,t}$ as the innovation vector $Z_{i,t}$ in~\eqref{alg-1}-\eqref{alg-2}, we introduce the log-moment generating function $\Lambda_M$ of $L_{i,t}$ at node $i$, when the measurements are drawn from $f_M$ (hypothesis $\mathbf H_M$ is true):
\begin{align}
\label{def-Lambda-M}
\Lambda_M(\lambda) & = \mathbb E\left[ \left.e^{\lambda^\top L_{i,t}}\right|\mathbf H=\mathbf H_{M}\right]\\
& = \mathbb E\left[ \left.e^{\sum_{m=1}^{M-1} \lambda_m \log\frac{f_{m}(Y_{i,t})}{f_{M}(Y_{i,t})}}\right|\mathbf H=\mathbf H_{M}\right],    \end{align}
for $\lambda = (\lambda_1,...,\lambda_{M-1})^\top\in \mathbb R^{M-1}$; we note that index $M$ in $\Lambda_M$ indicates the dependence on the assumed true distribution $f_M$. Similarly as in Section~\ref{sec-LD-metric}, the conjugate of $\Lambda_M$ is denoted by $I_M$. We assume that $\Lambda_M$ satisfies Assumption~\ref{ass-finite-at-all-points}. 

\subsection{Large deviations rates of the belief log-ratios}

The following result follows as a direct application of Theorem~\ref{theorem-nice-tight-bounds} to the log-ratios $X_{i,t}$ of public beliefs, defined in Section~\ref{subsec-SL}, eq.~\eqref{eq-X-i-t-SL}, $X_{i,t}^m = \frac{1}{t}\log \frac{b_{i,t}^m}{b_{i,t}^M}$, $m=1,...,M-1$. 
\begin{theorem}
Consider the social learning algorithm~\eqref{alg-SL-1}-\eqref{alg-SL-2} under Assumptions~\ref{ass-network-and-observation-model} and~\ref{ass-finite-at-all-points}, for $\Lambda=\Lambda_M$. Then, when $\mathbf H=\mathbf H_M$, for each node~$i$, for any measurable set $D$,
\begin{enumerate}
\mathitem
\begin{align}
\label{eq-UB-SL}
& \limsup_{t\rightarrow +\infty}\,\frac{1}{t}\,\log \mathbb P\left(X_{i,t}\in D\right) \leq - \inf_{x\in \overline D} I^\star_M (x),
\end{align}
where $I^\star_M (x) = \overline{\mathrm {co}}\,\inf\left\{I_M(x)+\mathcal J, N I_M(x)\right\}$;
\item for any collection $\mathcal H$ of graphs on $V$:
\begin{align}
\label{eq-LB-SL}
& \liminf_{t\rightarrow +\infty}\,\frac{1}{t}\,\log \mathbb P\left(X_{i,t}\in D\right) \geq - \inf_{x\in D^{\mathrm{o}}} I_{i,\mathcal H; M}(x),
\end{align}
where $I_{i,\mathcal H; M} (x) = \overline{\mathrm {co}}\,\inf\left\{|C_{i,\mathcal H}|I_M(x)+|\log p_{\mathcal H}|, N I_M(x)\right\}$.
\end{enumerate}
\end{theorem}
Consequently, all considerations, corollaries and results from Section~\ref{sec-Main} also carry over without any changes for the log-ratios $X_{i,t}$ of beliefs  in social learning. In particular, the LDP results for regular networks and pendant nodes also carry over to the social learning setup. 

\begin{theorem} [Almost sure convergence of $X_{i,t}$ in social learning]
\label{theorem-almost-sure-convergence-SL}
Consider the social learning algorithm~\eqref{alg-SL-1}-\eqref{alg-SL-2} under Assumptions~\ref{ass-network-and-observation-model} and~\ref{ass-finite-at-all-points}, for $\Lambda=\Lambda_M$. Then, for each node~$i$, for each $m=1,...,M-1$, $\frac{1}{t}\log \frac{b_{i,t}^m}{b_{i,t}^M}$ converges almost surely to $-D_{KL}(f_M||f_m) = - \mathbb E\left[\left. \log \frac{f_m(Y_{i,t})}{f_M(Y_{i,t}}\right| \mathbf H=\mathbf H_M\right].$
\end{theorem}
The result follows as a direct application of Theorem~\ref{theorem-almost-sure-convergence} for the case when the innovations $Z_{i,t}$ in~\eqref{alg-1}-\eqref{alg-2} are instantiated by the log-likelihood ratios $L_{i,t}$  defined in~\eqref{eq-Z-i-t-SL}, $L_{i,t}^m = \log \frac{f_m(Y_{i,t})}{f_M(Y_{i,t}}$, for $m=1,...,M-1$,  and by recognizing that the expected value of $ \log \frac{f_m(Y_{i,t})}{f_M(Y_{i,t}}$ under distribution $f_M$ is the negative of the KL divergence between $f_m$ and $f_M$.

 To illustrate the setup and the relevant quantities, we consider the example of $M$ scalar Gaussian distributions of different mean values and equal variances.
 
\begin{example}[Gaussian case: different mean values and equal variances]  
\label{example-Guassian-diff-means-equal-vars} 
Let $Y_{i,t}$ be Gaussian scalars, with mean value $\mu_m$ under hypothesis $m$, and (equal) variance $\sigma^2$. It is easy to show that, for this case, $L_{i,t}$ is computed as:
\begin{equation}
\label{eq-L-i-t-Gaussian}
L_{i,t} =\frac{1}{\sigma^2} \left(Y_{i,t}-\mu_M\right) d - D_{\mathrm{KL}}, 
\end{equation}
where $d=(d_1,...,d_{M-1})^\top$, and each $d_m = \mu_m-\mu_M$ is the difference between the mean value for the $m$-th hypothesis and the mean value for the true hypothesis, and $D_{\mathrm{KL}} = (D_{\mathrm{KL},1},...,D_{\mathrm{KL},M-1})^\top$, where $D_{\mathrm{KL},m}= \frac{(\mu_m-\mu_M)^2}{2\sigma^2}$ is the KL divergence between the distribution $f_m$ and the true distribution $f_M$, $m=1,...,M-1$. It is easy to see that, for each $i=1,...,N$ and each $t$,  $L_{i,t}$ is Gaussian with mean vector $-D_{\mathrm{KL}}$ and covariance matrix $\frac{1}{\sigma^2} d d^\top$. Using the standard formula for the log-moment generating function of multivariate Gaussian distribution, we get:
\begin{equation}
\Lambda_M(\lambda) = - \lambda^\top D_{\mathrm{KL}} + 
\frac{(\lambda^\top d)}{2 \sigma^2}.  
\end{equation}
Simple calculus shows that the conjugate function $I_M$ is given by:
\begin{equation}
I_M(x) = \left\{
\begin{array}{cc}
\frac{\zeta^2}{2\sigma^2},     & \mathrm{if\,} x = \frac{\zeta}{2\sigma^2} d - D_{\mathrm{KL}},\mathrm{\,for\,some\,}\zeta\in \mathbb R\\
+\infty,     & \mathrm{if\,}x+D_{\mathrm{KL}}\notin \mathrm{span}(d)
\end{array}.
\right.    
\end{equation}
Thus, $I_M$ is essentially a one-dimensional quadratic function that changes only along the direction $-D_{KL} + \alpha d$, $\alpha \in \mathbb R$, while being equal to $+\infty$ in the rest of the $\mathbb R^d$ space. This is intuitive as the log-likelihood ratios for different $m$ are coupled through a common (scalar) variable $Y_{i,t},$  and hence the events that vector $L_{i,t}$ lies outside of the line $-D_{KL} + \alpha d$ must have zero probability (and thus rate function equal to $+\infty$). The convex conjugates of $I_M$ from Theorem~\ref{theorem-nice-tight-bounds},  $I^\star_M$ and $I_{i,\mathcal H; M}$, can be found similarly as in Section~\ref{subsec-a-closer-look-at-conjugate-functions-from-the-theorem}.   
\end{example}

\subsection{Large deviations rates for beliefs in social learning}
For each $m=1,..,M-1$, define $g_m:\mathbb R^{M-1} \mapsto \mathbb R$ as $g_m(x)= x_m- \max\{0, x_1,...,x_{M-1}\},$ for $x\in \mathbb R^d$.
\begin{theorem}
\label{theorem-LD-rates-for-beliefs-in-SL}
Consider the social learning algorithm~\eqref{alg-SL-1}-\eqref{alg-SL-2} under Assumptions~\ref{ass-network-and-observation-model} and~\ref{ass-finite-at-all-points}, for $\Lambda=\Lambda_M$. Then, for each node~$i\in V$ and hypothesis $m=1,...,M-1$, for any given interval $F\subseteq \mathbb R$: 
\begin{enumerate}
\item 
\begin{align}
\label{eq-UB-SL-beliefs}
& \limsup_{t\rightarrow +\infty}\,\frac{1}{t}\,\log \mathbb P\left( \frac{1}{t}\log b_{i,t}^m \in F\right) \leq - \inf_{x: g_m(x) \in F} I^\star_M (x);
\end{align}
\item for any disconnected collection $\mathcal H$,  
\begin{align}
\label{eq-LB-SL-beliefs}
& \liminf_{t\rightarrow +\infty}\,\frac{1}{t}\,\log \mathbb P\left(\frac{1}{t}\log b_{i,t}^m \in F\right) \geq - \inf_{x: g_m(x)\in F} I_{i,\mathcal H; M}(x).
\end{align}
\end{enumerate}

\end{theorem}
The proof is very similar to the proof of Lemma~4 from~\cite{Lalitha18SLandDHT}. The key distinction is that here full LDP for the log-ratios of the beliefs, $X_{i,t}$, is not available due to the complexity of the network model, and we have to work instead with the upper and the lower rate function bounds. 
However, the key arguments remain unaltered. For completeness, we provide the main steps of the proof in Appendix~\ref{app:lemma-rates-for-SL-beliefs}. 

The result in Theorem~\ref{theorem-LD-rates-for-beliefs-in-SL} is very general, as it holds for arbitrary distributions $f_m$, $m=1,...,M$, such that the log-moment generating function $\Lambda_M$ satisfies Assumption~\ref{ass-finite-at-all-points}; this is for example the case for Gaussian distributions from Example~\ref{example-Guassian-diff-means-equal-vars}. 

\begin{remark}
It can be shown by carrying out the same analyses as in the proof of Theorem~\ref{theorem-LD-rates-for-beliefs-in-SL}, that, if for some node $i$ the sequence $X_{i,t}$ satisfies the LDP with rate function $I_i$, then, for each $m=1,...,M$ the sequence of log beliefs $\frac{1}{t}\log b_{i,t}^m$ also satisfies the LDP with rate function
\begin{equation}
\label{eq-rate-function-log-beliefs-DL}
R_{i,m}(z) = \inf_{x:\, g_m(x)=z} I_i(x),    
\end{equation}
for $x\in \mathbb R$.
\end{remark} 

We can see that, to find the large deviations rates of the beliefs, first the rate function $I_i$ (or bounds on this function) for the log-belief ratios $X_{i,t}$ are found, and then the contraction principle is applied with functions $g_m$ acting as the bridge between the two domains. This relation is established in~\cite{Lalitha18SLandDHT} for static networks, but the same behaviour carries over to the general case, with the difference that the rate function of log-beliefs can differ across different nodes as a result of network randomness. To shed some light on function $g_m$, we revisit Example~\ref{example-Guassian-diff-means-equal-vars} for which we derive a closed form expression for $g_m$.  

\begin{example}[Computation of $g_m$ for the Gaussian case]
\label{ex-Gaussian-g-m}
Consider the setup from Example~\ref{example-Guassian-diff-means-equal-vars}. Recall that $g_m(x) = x_m - \max\{0, x_1,...,x_{M-1}\}$ and also that $I_M(x)=+\infty$ outside of the line $- D_{KL} + \alpha d,$ $\alpha \in \mathbb R$. Define also
\begin{equation}
\label{eq-def-f}
f(\zeta) = \max \{0, \zeta \frac{d_1}{\sigma^2} - D_{KL,1},..., \zeta \frac{d_{M-1}}{\sigma^2} - D_{KL,M-1}\}.     
\end{equation}
and note that 
\begin{equation}
\label{eq-g-m-f-relation}
 g_m(x) = \zeta \frac{d_m}{\sigma^2} - D_{KL,m} - f(\zeta),   
\end{equation}
for any $\zeta \in \mathbb R$ and $x \in \mathbb  R^{M-1}$ such that $x=(\zeta \frac{d_1}{\sigma^2} - D_{KL,1},..., \zeta \frac{d_{M-1}}{\sigma^2} - D_{KL,M-1})$, for $m=1,...,M-1$.  

Without loss of generality, assume that $\mu_1<\mu_2<...<\mu_{M-1}$, implying also $d_1<d_2<...<d_{M-1}$. Let $m^\star$ be the largest $m$ such that $\mu_m < \mu_M$, $m^\star = \max \{m \in \{0,1,...,M-1\}: \mu_m<\mu_M\}$, wherein we additionally define $\mu_0 \equiv-\infty$ to account for the case that $\mu_M<\mu_1$. Then $d_1<...<d_m^\star < 0 < d_{m^\star+1}<...<d_{M-1}$. By the preceding ordering, and exploiting also that $D_{KL,m} = \frac{d_m^2}{2\sigma^2}$, it can be easily verified that for any pair $l<m$, the intersection between the lines $\zeta \frac{d_l}{\sigma^2} - D_{KL,l}$ and $\zeta \frac{d_m}{\sigma^2} - D_{KL,m}$ occurs at $\frac{d_l+d_m}{2}$, with the $l$-indexed line dominating to the left of this point, for $\zeta < \frac{d_l+d_m}{2}$, while the $m$-indexed line dominates to the right. It also clearly follows that the first intersection point occurs for the first neighboring index, thus, as $\zeta$ increases, the lines must dominate in the same order as their $d_m$ values. Summarizing, $f$ is given in the following form:     
\begin{equation}
\label{eq-f-closed-form}
f(\zeta) = \left\{\begin{array}{cc}
    \zeta \frac{d_1}{\sigma^2} - D_{KL,1}, &  \zeta < \frac{d_1+d_2}{2} \\
    \zeta \frac{d_2}{\sigma^2} - D_{KL,2},  &  \frac{d_1+d_2}{2} \leq \zeta < \frac{d_2+d_3}{2} \\
     ... & \\
     0 & \frac{d_{m^\star}}{2} \leq \zeta < \frac{d_{m^\star+1}}{2} \\
     ... & \\
\zeta \frac{d_m}{\sigma^2} - D_{KL,m},  &  \frac{d_{m-1}+d_m}{2} \leq \zeta < \frac{d_m+d_{m+1}}{2} \\
...& \\     
     \zeta \frac{d_{M-1}}{\sigma^2} - D_{KL,M-1}, &  \zeta \geq \frac{d_{M-2}+d_{M-1}}{2}
\end{array}.  \right.  
\end{equation}
From~\eqref{eq-g-m-f-relation} and~\eqref{eq-f-closed-form}, we can obtain for $x = \frac{\zeta}{\sigma^2} d -D_{KL}$:
\begin{equation}
\label{eq-g-m-closed-form}
g_m(x) = \left\{\begin{array}{cc}
    \frac{(d_m-d_1)}{\sigma^2}\left(\zeta - \frac{d_m+d_1}{2}\right)  &  \zeta < \frac{d_1+d_2}{2} \\
    \zeta \frac{d_2}{\sigma^2} - D_{KL,2}  &  \frac{d_1+d_2}{2} \leq \zeta < \frac{d_2+d_3}{2} \\
     ... & \\
\zeta \frac{d_m}{\sigma^2} - D_{KL,m} & \frac{d_{m^\star}}{2} \leq \zeta < \frac{d_{m^\star+1}}{2} \\
     ... & \\
     0 & \frac{d_{m-1}+ d_{m}}{2} \leq \zeta < \frac{d_{m}+d_{m +1}}{2} \\
     ... & \\
     \frac{(d_m-d_{M-1})}{\sigma^2}\left(\zeta - \frac{d_m+d_{M-1}}{2}\right)  &  \zeta \geq \frac{d_{M-2}+d_{M-1}}{2}
\end{array}.  \right.  
\end{equation}
The derived closed form expression for $g_m$ is a step towards deriving the closed form expression for the rate function $R_{i,m}$, and, in particular, it suggests an analytical validation for the piece-wise behaviour of the rate function of beliefs discovered numerically in~\cite{Lalitha18SLandDHT}, Figure 9. This is out of scope of the current paper and is left for future work. To provide an illustration towards characterizing $R_{i,m}$, we consider the value of the rate function at $-D_{KL,m}$.  From~\eqref{eq-g-m-closed-form}, we see that $g_m(x) = -D_{KL,m}\}$ for $x=\frac{\zeta}{\sigma^2} - D_{KL}$ and $\zeta = 0$ (note that, by construction, $d_{m^\star}<0$ and $d_{m^\star+1}>0$, and hence $\zeta=0 \in [\frac{d_{m^\star}}{2}, \frac{d_{m^\star+1}}{2})$). Thus, we have 
\begin{equation}
R_{i,m}(-D_{KL,m}) = \inf_{x: g_m(x) = -D_{KL,m}} I_i(x) \leq I_i(-D_{KL}).
\end{equation}
the preceding inequality holds trivially by the fact that $-D_{KL} \in \{x: g_m(x) = -D_{KL,m}\}.$ On the other hand, we have proved that $I^\star_{M}(-D_{KL}) = I_{i,\mathcal H; M}(-D_{KL})=0$ (see Remark~\ref{remark-zero-rate-at-theta}). By~\eqref{eq-sandwich-I-i}, we thus have $I_{i}(-D_{KL})=0$. It follows that $R_{i,m}(-D_{KL,m})=0$, i.e., the derived expression for $g_m$ reveals that the value of the rate function $R_{i,m}$ at $-D_{KL,m}$ is zero. This is in accordance with almost sure convergence of $\frac{1}{t}\log b_{i,t}^m$ to $-D_{KL,m}$ which follows by combining Theorems 26 and 32.   
\end{example}

When the two functions from~\eqref{eq-UB-SL} and~\eqref{eq-LB-SL}, namely, $I^\star_{M}$ and $I_{i,\mathcal H;M}$ match, this implies that the corresponding $\limsup$ and the $\liminf$ are equal. Hence, whenever for a given node $i$ its sequence $X_{i,t}$ exhibits LDP, this implies LDP for the sequence of beliefs $\frac{1}{t}\log b_{i,t}^m$, for each $m=1,...,M-1$. Here we give an example for regular networks.  
\begin{corollary}[LDP for social learning in regular networks]
Suppose that $\overline G$ is a circulant network as in Corollary~\ref{cor-regular}, i.e.,  each node is connected to $d/2$ nodes on the left and $d/2$ nodes on the right, where $d\leq N-1$ is even. We assume that each link, independently of all other links, fails with probability $1-p$. Then, for any node $i$, for each $m$,  $\frac{1}{t}\log b_{i,t}^m$ satisfies the LDP with the rate function 
\begin{equation}
R_m(z) = \inf_{x\in \mathbb R^{M-1}: g_m(x)= z} \co\left\{ N I_M, I_M+ d \log \left|1-p\right| \right\}(x).    
\end{equation}
\end{corollary}
A similar result holds also for pendant nodes with i.i.d. link failures. 

\mypar{Convergence to the correct hypothesis}
The next result establishes, through the use of large deviations analysis, that the social learning algorithm~\eqref{alg-SL-1}-\eqref{alg-SL-2} correctly identifies the true hypothesis. We remark that this recovers the result of~\cite{Parasnis20} for the special case of identical distributions across nodes. 
\begin{theorem} 
Consider the social learning algorithm~\eqref{alg-SL-1}-\eqref{alg-SL-2} under Assumption~\ref{ass-network-and-observation-model} and~\ref{ass-finite-at-all-points}, for $\Lambda=\Lambda_M$. Then, when $\mathbf H=\mathbf H_M$, for each node $i$, the sequence of beliefs $b_{i,t}^M$ converges to one almost surely. 
\end{theorem}

\begin{proof}
From the construction of the beliefs $b_{i,t}^m$, for each $i,t$, $b_{i,t}^M= 1- b_{i,t}^1 - ... - b_{i,t}^{M-1}$. Combining this with the relations $X_{i,t}^m = \frac{1}{t}\log \frac{b_{i,t}^m}{b_{i,t}^M}$, yields 
\begin{equation}
b_{i,t}^M= \frac{1}{1+ \sum_{m=1}^{M-1} e^{t X_{i,t}^m}}.     
\end{equation}
By Theorem~\ref{theorem-almost-sure-convergence-SL}, for each $m=1,...,M-1$, $X_{i,t}^m$ converges almost surely to $-D_{KL}(f_M||f_m)<0$. Hence, each of the terms $e^{t X_{i,t}^m}$ in the sum above vanishes with probability one. Since $M$ is finite, there exists a set of probability one such that $\sum_{m=1}^M e^{t X_{i,t}^m}$ vanishes, proving that $b_{i,t}^M$ converges to one almost surely.   
\end{proof}

The next two sections prove Theorem~\ref{theorem-nice-tight-bounds};
Section~\ref{subsection-nice-upper-bound} proves the upper bound~\eqref{eq-UB-Theorem}
and Section~\ref{subsection-nice-lower-bound} proves the lower bound~\eqref{eq-LB-Theorem}.

\section{Proof of Theorem~\ref{theorem-nice-tight-bounds}}
\label{section-proofs}
This section proves Theorem~\ref{theorem-nice-tight-bounds} by proving separately the upper and the lower bound. Before giving the respective proofs, we first give some important lemmas that are used in both the upper and the lower bound proof.

Lemma~\ref{simple-lemma} will be used to find the log-moment generating function of the estimate $X_{i,t}$ from the log-moment generating functions of each of the terms in the sum~\eqref{alg-compact}. This result follows from convexity and zero value at the origin property of $\Lambda$.
\begin{lemma}
\label{simple-lemma}
For any set of convex multipliers $\alpha \in \Delta_{N-1}$, for each $j=1,...,N$, the log-moment generating function $\Lambda$ satisfies,
\begin{equation}
\label{eq-simple-lemma}
N \Lambda \left(1/N \lambda\right)\leq \sum_{i=1}^N \Lambda(\alpha_i \lambda) \leq \Lambda(\lambda),
\end{equation}
for any $\lambda \in \mathbb R^d$.
\end{lemma}
The proof of Lemma~\ref{simple-lemma} can be found in~\cite{DI-Directed-Networks16}.

The claims in Lemma~\ref{lemma-conjugate-calculus} are standard results from convex analysis, the proofs of which can be found, e.g., in~\cite{Urruty}. Let the superscript $^\star$ denote the conjugacy operation, i.e., for a given function $f:\mathbb R^d\mapsto\mathbb R$,
\begin{equation}
f^\star(x)=\sup_{s\in \mathbb R^d} s^\top x - f(s),\;\;\;x\in \mathbb R^d.
\end{equation}
The following relations hold between a function $f$ and its conjugate $f^\star$.
\begin{lemma}
\label{lemma-conjugate-calculus}
\begin{enumerate}
\item \label{part-basic} Let $f:\mathbb R^d\mapsto \mathbb R$ be a given a function. Then:
\begin{enumerate}
\item \label{part-conjugate-shifting}
$[f(\cdot)+r]^\star= f^\star(\cdot)-r$;
\item \label{part-conjugate-scaling}
for $\alpha>0$ and $\beta\neq 0$, $\left[\alpha f(\beta (\cdot))\right]^\star = \alpha f^\star\left(1/(\alpha \beta) (\cdot)\right)$.
\end{enumerate}
\item \label{part-conj-of-the-max}
Let $f_1$ and $f_2$ be two given functions. Then, the conjugate of the pointwise supremum of $f_1$ and $f_2$ is the convex hull of the pointwise infimum of $f_1^\star$ and $f_2^\star$:
\begin{equation}
[\sup\{f_1,f_2\}]^\star=\overline {\mathrm{co} }\inf\left\{ f_1^\star, f_2^\star\right\}.
\end{equation}
\end{enumerate}
\end{lemma}

\subsection{Proof of the upper bound~\eqref{eq-UB-Theorem}}
\label{subsection-nice-upper-bound}
In our previous work~\cite{DI-Directed-Networks16}, we have proved that, at any node $i$, the sequence $X_{i,t}$ is exponentially tight. This intuitively means that the probabilities of the tail events of $X_{i,t}$ vanish sufficiently fast (i.e., the exponential rates of the tail probabilities grow unbounded when the tails move to infinity). Lemma~\ref{lemma-UB-skeleton-d>1} uses this result to derive an elegant sufficient condition for a certain function to satisfy the large deviations upper bound from Definition~\ref{def-LDP}. In our case, this function will be the conjugate of a certain modification of $\Lambda$ that accounts for the effects of intermittent communications. We remark that at the core of the proof of Lemma~\ref{lemma-UB-skeleton-d>1} is a modification of the finite cover argument from the proof of Cram\'er's theorem in $\mathbb R^d$ (see, e.g.,~\cite{DemboZeitouni93}); the detailed proof of Lemma~\ref{lemma-UB-skeleton-d>1} is provided in Appendix~\ref{app:UB-skeleton}.
\begin{lemma}
\label{lemma-UB-skeleton-d>1}
Let $X_t$ be an arbitrary sequence of random variables where each $X_t$ takes values in $\mathbb R^d$. Suppose that for some function $f$, for any measurable set $D$ there holds
\begin{equation}
\label{eq-if-UB}
\limsup_{t\rightarrow +\infty} \,\frac{1}{t}\,\log \mathbb P\left(X_t\in D\right) \leq f(\lambda) - \inf_{x\in D} \lambda^\top x,
\end{equation}
for any $\lambda\in \mathbb R^d$. Then, if $f$ is finite for all $\lambda \in \mathbb R^d$, for any compact set $F$
\begin{equation}
\label{eq-then-UB}
\limsup_{t\rightarrow +\infty} \,\frac{1}{t}\,\log \mathbb P\left(X_t\in F\right) \leq -\inf_{x\in F} f^\star(x),
\end{equation}
where $f^\star$ is the conjugate of $f$. If in addition $X_t$ is exponentially tight, then~\eqref{eq-then-UB} holds for any closed set $F$.
\end{lemma}

Fix an arbitrary node $i\in V$. Replicating the steps of the proof of Theorem~5 from~\cite{Non-Gaussian-DD}, we obtain that, for any measurable set $D$, and any fixed $\lambda\in \mathbb R^d$,
\begin{align}
&\limsup_{t\rightarrow+\infty} \frac{1}{t}\log\mathbb P\left(X_{i,t}\in D\right)\nonumber\\
&\leq \max\left\{N\Lambda\left(\frac{1}{N}\lambda\right), \Lambda(\lambda)-\mathcal J\right\} - \inf_{x\in D}\lambda^\top x.
\end{align}

By Lemma~17 from~\cite{DI-Directed-Networks16}, the sequence of estimates $X_{i,t}$ is exponentially tight. (We remark that this result is proven under more general assumptions on the weight matrices than assumed here.) Hence, to prove the upper bound~\eqref{eq-UB-Theorem}, it only remains to show that $I^\star$ from Theorem~\ref{theorem-nice-tight-bounds} is the conjugate of $f (\lambda) := \max\left\{N\Lambda\left(1/N\lambda\right), \Lambda(\lambda)-\mathcal J\right\}$, $\lambda\in \mathbb R^d.$ From part~\ref{part-conj-of-the-max} of Lemma~\ref{lemma-conjugate-calculus}, we have that the conjugate of $f$ is the closed convex hull of the infimum of the conjugates of $f_1(\lambda):=\lambda\mapsto N\Lambda\left(1/N\lambda\right)$ and $f_2(\lambda):=\lambda\mapsto\Lambda(\lambda)-\mathcal J$. Using the conjugacy rules from parts~\ref{part-conjugate-scaling} and~\ref{part-conjugate-shifting} of Lemma~\ref{lemma-conjugate-calculus}, we obtain that the respective conjugates of $f_1$ and $f_2$ are $N I(x)$, $x\in \mathbb R^d$, and $I(x)+\mathcal J$, $x\in \mathbb R^d$. The upper bound~\ref{eq-UB-Theorem} follows by part 2 of Lemma~\ref{lemma-conjugate-calculus}.

\subsection{Proof of the lower bound~\eqref{eq-LB-Theorem}}
\label{subsection-nice-lower-bound}
Fix an arbitrary node $i\in V$. Fix a collection of feasible graphs $\mathcal H$. To simplify the notation, we denote the component of $i$ in $\mathcal H$, $C_{i,\mathcal H}$, by $C$. We also let $M$ denote the number of nodes in $C$, $M=|C|$. 
For each fixed $t$, we define the family of events $\left\{\mathcal E^t_{\theta}: \theta \in [0,1]\right\}$, such that for any $\theta\in [0,1]$,
\begin{align}
\label{eqn-def-of-sets-E}
\mathcal E^t_{\theta}&=\left\{  G_s \in \mathcal H,\,\lceil \theta t \rceil \leq  s \leq t,
\;\; \left\|[\Phi(t,t-o_t)]_C -J_M   \right\|\leq \frac{1}{t},\;\right.\nonumber\\
 & \left.\;\;\;\;\;\;\;\;\;\;\;\;\;\;\;\;\;\; \left\|\Phi(\lceil \theta t \rceil, \lceil \theta t \rceil- o_t )
  -J_N   \right\|  \leq \frac{1}{t} \right\},
\end{align}
where $o_t= \left\lceil \log t \right\rceil$; we recall that, for a square matrix $A$, $A_C$ denotes the block of $A$ corresponding to the intersection of columns and rows of $A$ the indices of which belong to $C$. For convenience, we introduce $\mathcal T_{\theta}=\left\{\lceil \theta t \rceil,...,t\right\}$.
\begin{lemma}
\label{lemma-blocks}
Let $\theta$ be an arbitrary number in $[0,1]$. For any $\omega \in \mathcal E^t_{\theta}$,
\begin{enumerate}
\item\label{lemma-blocks-part-1}
 for any $s\in \mathcal T_{\theta}$, \[[\Phi(t,s)]_{i j}=0, \mbox{\;\;for\;\;} j\notin C;\]
\item \label{lemma-blocks-part-2}
for $t-o_t \geq s\geq \lceil \theta t \rceil,$
\[\left|[\Phi(t,s)]_{ij} -\frac{1}{M}\right| \leq \frac{1}{t}, \mbox{\;\;for\;all\;\;} j\in C;\]
\item \label{lemma-blocks-part-3}
for $\lceil \theta t \rceil-o_t \geq s\geq 1,$
\[\left|[\Phi(t,s)]_{ij} -\frac{1}{N}\right| \leq \frac{1}{t}, \mbox{\;\;for\;all\;\;} j \in V.\]
\end{enumerate}
\end{lemma}
\begin{proof}
Fix $\omega \in \mathcal E^t_{\theta}$ and, for $s=1,...,t$, denote $A_s=W_s(\omega)$. Consider first part~\ref{lemma-blocks-part-1}, and suppose, without loss of generality, that $C=\{ 1,...,M\}$. By construction of $\mathcal E^t_{\theta}$, none of the graphs that appear during $\mathcal T_{\theta}$ have links that connect $C$ with the remaining part of the network $C^{\mathrm{c}}=V\setminus C$. Hence, each of the matrices $A_s$, $s\in \mathcal T_{\theta}$ has the following block diagonal form
\begin{equation}
\label{eq-blk-diagonal}
A_s=\left[ \begin{array}{cc}
[A_s]_C & 0_{M\times (N-M)}\\
0_{M\times (N-M)} & [A_s]_{V\setminus C} \end{array}
\right],
\end{equation}
and the same structure is therefore preserved in their products $\Phi(t,s)=A_t \cdots A_s$, $s\in \mathcal T_\theta$, i.e.,
\begin{equation*}
\Phi(t,s)=\left[ \begin{array}{cc}
[A_t]_C \cdot \ldots \cdot [A_s]_C & 0_{M\times (N-M)}\\
0_{M\times (N-M)} & {[A_t]}_{C^{\mathrm{c}}} \cdot \ldots \cdot {[A_s]}_{C^{\mathrm{c}}}\end{array}
\right].
\end{equation*}

We next consider part~\ref{lemma-blocks-part-2}.  Since for an arbitrary matrix $A$, for any $i,j$ there holds $|A_{ij}|\leq \|A\|$, it is sufficient to show that $\left\| [\Phi(t,s)]_C -J_M \right\| \leq 1/t$, for any fixed $s\in \mathcal T_{\theta}$ such that $s \leq t-o_t$. By part~\ref{lemma-blocks-part-1}, we know that for any $s_1,s_2\in \mathcal T_{\theta}$, the $C$ block of $\Phi(s_1,s_2)$ is computed as the product of blocks  $[A_{s_1}]_C$ through $[A_{s_2}]_C$. Since each of these blocks is a symmetric, stochastic, $M$ by $M$ matrix, we have that $[\Phi(s_1,s_2)]_C$ is a doubly stochastic ($M$ by $M$) matrix. Consider now a fixed $s \in \mathcal T_{\theta}$ such that $s\leq t-o_t$. Factoring out $[\Phi(t,s)]_C$ as the product $[\Phi(t,t-o_t)]_C \Phi(t-o_t-1,s)]_C$, and using the double-stochasticity of the latter two matrices, we obtain $[\Phi(t,s)]_C-J_M= ([\Phi(t,t-o_t)]_C-J_M) (\Phi(t-o_t-1,s)]_C-J_M)$. By construction of $\mathcal E^t_{\theta}$, the spectral norm of the first factor is not greater than $1/t$, while the double-stochasticity of $\Phi(t-o_t-1,s)]_C$ yields that the spectral norm of the second factor is not greater than $1$. Using submultiplicativity of the spectral norm, the claim in part~\ref{lemma-blocks-part-2} follows:
\begin{align}
& \left\| [\Phi(t,s)]_C -J_M \right\|  \nonumber\\
& \leq \left\|  [\Phi(t,t-o_t)]_C - J_M\right\| \left\|[\Phi(t-o_t-1,s)]_C - J_M\right\|\nonumber\\
& \leq 1/t.
\end{align}

Part~\ref{lemma-blocks-part-3} can be proven by factoring out $\Phi(t,s)$ as the product $\Phi(t,\lceil \theta t\rceil) \Phi(\lceil \theta t\rceil-1,\lceil \theta t\rceil-o_t)\Phi(\lceil \theta t\rceil-o_t-1,s)$ and applying similar arguments as in the proof of part~\ref{lemma-blocks-part-2}.
\end{proof}

Fix $\theta\in [0,1]$ and consider the probability distribution $\nu^{\theta}_t: \mathcal B\left( \mathbb R^d\right)\rightarrow [0,1]$ defined by
 \begin{equation}
 \nu^{\theta}_t (D)=  \frac{\mathbb P\left( \left\{X_{i,t}\in D\right\} \cap \mathcal E^t_{\theta}\right)}
 {\mathbb P\left(\mathcal E^t_{\theta}\right)},
 \end{equation}
that is, $\nu^{\theta}_t$ is the probability distribution of $X_{i,t}$ conditioned on the event $\mathcal E^t_{\theta}$ (we note that $\mathbb P\left(\mathcal E^t_{\theta}\right)>0$ for $t$ sufficiently large, as we show later in the proof, see Lemma~\ref{lemma-probability-of-mathcalE} further ahead).

Let $\Upsilon_t$ be the (normalized) logarithmic moment generating function associated with $\nu^{\theta}_t$,
\begin{equation}
\label{eqn-Upsilon-t}
\Upsilon_t(\lambda)=\frac{1}{t}\log
\mathbb E\left[ e^{t \lambda^\top X_{i,t}} \left| \mathcal E^t_{\theta}\right. \right],\mathrm{\;\;for\;\;}\lambda\in \mathbb R^d.
\end{equation}

Using the properties of entries of $\Phi(t,s)$ for different intervals on $s$ listed in Lemma~\ref{lemma-blocks}, we establish in Lemma~\ref{lemma-sequence-Upsilon-t-has-a-limit} that the sequence of functions $\Upsilon_t$ has a point-wise limit for every $\lambda \in \mathbb R^d$. This will allow to apply the G\"{a}rtner-Ellis theorem~\cite{DemboZeitouni93} to compute the large deviations rate function for the sequence of measures $\nu^{\theta}_t$. We first state and prove Lemma~\ref{lemma-sequence-Upsilon-t-has-a-limit}.

\begin{lemma}
\label{lemma-sequence-Upsilon-t-has-a-limit}
For any $\lambda\in \mathbb R^d$ and any $\theta \in [0,1]$:
\begin{equation}
\lim_{t\rightarrow +\infty} \Upsilon_t(\lambda)= (1-\theta) M\Lambda\left(\frac{1}{M} \lambda\right) +
\theta N \Lambda\left(\frac{1}{N} \lambda\right),
\end{equation}
where, we recall, $M=|C|$.
\end{lemma}
\begin{proof}
Fix $\theta \in [0,1]$, $\lambda\in \mathbb R^d$. We have:
\begin{align}
\label{eqn-conditioning-on-matrices-in-mathcalE}
\mathbb E\left[ e^{t \lambda^\top X_{i,t}}\left| \mathcal E^t_{\theta}\right. \right]
&= \frac{1}{\mathbb P\left(  \mathcal E^t_{\theta} \right) }
\mathbb E\left[ 1_{\mathcal E^t_{\theta}}  e^{t \lambda^\top X_{i,t}} \right] \nonumber\\
&= \frac{1}{\mathbb P\left(  \mathcal E^t_{\theta} \right) }
\mathbb E\left[\mathbb E\left[ 1_{\mathcal E^t_{\theta}}  e^{t \lambda^\top X_{i,t}} | W_1,...,W_t\right]\right]\nonumber\\
&=\frac{1}{\mathbb P\left(  \mathcal E^t_{\theta} \right) }
\mathbb E\left[  1_{\mathcal E^t_{\theta}}\mathbb E\left[e^{t \lambda^\top X_{i,t}} | W_1,...,W_t\right]\right],
\end{align}
where in the last equality we used that the indicator $1_{\mathcal E^t_{\theta}}$ is a function of $W_1,...,W_t$.
Further, as the summands in~\eqref{alg-compact} are independent given $W_1,...,W_t$, we obtain
\begin{equation}
\label{eqn-condition-on-fixed-A1-Ak}
\mathbb E\left[ e^{t \lambda^\top X_{i,t}} |W_1,...,W_t \right]
=e^{ \sum_{s=1}^t \sum_{j=1}^N   \Lambda\left( [\Phi(t,s)]_{ij} \lambda\right) }.
\end{equation}
Consider now a fixed $\omega \in \mathcal E^t_{\theta}$. We split the sum in the exponent of~\eqref{eqn-condition-on-fixed-A1-Ak} according to the intervals used in the construction of $\mathcal E^t_{\theta}$. With this in mind, we define also
\begin{align}
\overline \chi_t & := \max_{\alpha \in [1/M-1/t, 1/M+1/t]} \Lambda\left(\alpha \lambda \right),
\label{eq-def-Chi-t-upper} \\
\underline \chi_t & := \min_{\alpha \in [1/M-1/t, 1/M+1/t]} \Lambda\left(\alpha \lambda \right),
\label{eq-def-Chi-t-lower}
\end{align}
and
\begin{align}
\overline \zeta_t &:= \max_{\alpha \in [1/N-1/t, 1/N+1/t]} \Lambda\left(\alpha \lambda \right),
\label{eq-def-Zeta-t-upper} \\
\underline \zeta_t &:= \min_{\alpha \in [1/N-1/t, 1/N+1/t]} \Lambda\left(\alpha \lambda \right),
\label{eq-def-Zeta-t-lower}
\end{align}
for $\lambda \in \mathbb R^d$. We remark that, by the continuity of $\Lambda$ and compactness of the intervals, in each of the preceding optimization problems there exists a maximizer. Further, as $t\rightarrow +\infty$, the corresponding intervals shrink to a single point, and by using again continuity of $\Lambda$, we obtain that $\overline \chi_t,\underline \chi_t\rightarrow \Lambda\left(1/M \lambda \right)$, and $\overline \zeta_t,\underline \zeta_t\rightarrow \Lambda\left(1/N \lambda \right)$, as $t\rightarrow +\infty$.  Then, by part~\ref{lemma-blocks-part-1} of Lemma~\ref{lemma-blocks} and the fact that $\Lambda(0)=0$, we have \begin{equation*}
\sum_{j\notin C} \Lambda\left( [\Phi(t,s)]_{ij} \lambda \right)=0, \mathrm{\;\;for\; each\;}  s\in \mathcal T_{\theta}.
\end{equation*}
Further, by part~\ref{lemma-blocks-part-2} of Lemma~\ref{lemma-blocks}
\begin{equation*}
M \underline \chi_t \leq  \sum_{j\in C} \Lambda\left( [\Phi(t,s)]_{ij} \lambda \right)
\leq   M \overline \chi_t,  \mathrm{\;\;for\;}  t-o_t\geq s\geq \lceil\theta t\rceil,
\end{equation*}
and, similarly, by part~\ref{lemma-blocks-part-3} of Lemma~\ref{lemma-blocks}
\begin{equation*}
N \underline \zeta_t \leq  \sum_{j=1}^N \Lambda\left( [\Phi(t,s)]_{ij} \lambda \right)
\leq   N \overline \zeta_t,  \mathrm{\;\;for\;}  \lceil\theta t\rceil-o_t\geq s\geq 1.
\end{equation*}
As for the summands in the intervals $\left\{t,...,t-o_t\right\}$ and $\left\{\lceil \theta t\rceil,...,\lceil \theta t\rceil-o_t\right\}$, we apply
Lemma~\ref{simple-lemma} to get
\begin{align*}
M \Lambda \left(\frac{1}{M} \lambda\right) \leq  \sum_{j\in C} \Lambda\left( [\Phi(t,s)]_{ij} \lambda \right)
& \leq   \Lambda(\lambda),\\
               & \mathrm{\;\;for\;} t\geq s\geq t-o_t,
\end{align*}
and
\begin{align*}
N \Lambda \left(1/N \lambda\right) \leq  \sum_{j=1}^N \Lambda\left( [\Phi(t,s)]_{ij} \lambda \right)
& \leq   \Lambda(\lambda),\\
               & \mathrm{\;\;for\;} \lceil \theta t\rceil\geq s\geq \lceil \theta t\rceil-o_t.
\end{align*}
Summing out the upper and lower bounds over all $s$ in the preceding five inequalities yields:
\begin{equation}
\label{eq-Upsilon-t-bounds}
  t\, \underline \Upsilon_t\left(\lambda \right)
   \leq  \sum_{s=1}^t \sum_{i=1}^N \Lambda\left([\Phi(t,s)]_{i,j}\right) \leq
  t\, \overline \Upsilon_t\left(\lambda \right),
    \end{equation}
where
\begin{align*}
\underline \Upsilon_t\left(\lambda \right) & = \frac{\lceil \theta t \rceil - o_t}{t} N \underline \zeta_t
 + \frac{o_t}{t} \left( N \Lambda\left( \frac{1}{N}\lambda \right)+ M  \Lambda\left( \frac{1}{M} \lambda \right)\right) \\
& + \frac{t- \lceil \theta t \rceil -  o_t}{t}  M \underline \chi_t,
\end{align*}
and
\begin{align*}
\overline \Upsilon_t(\lambda)& = \frac{\lceil \theta t \rceil - o_t}{t} N \overline \zeta_t
+ \frac{o_t}{t}\left( N \Lambda\left( \frac{1}{N}\lambda \right)+ M  \Lambda\left( \frac{1}{M}\lambda \right)\right)\\
& + \frac{t- \lceil \theta t \rceil -o_t}{t} M \overline \chi_t.
\end{align*}

The inequalities in~\eqref{eq-Upsilon-t-bounds} hold for any fixed $\omega \in \mathcal E^t_{\theta}$. Thus,
\begin{equation}
1_{\mathcal E^t_{\theta} } e^{ t \underline \Upsilon_t\left(\lambda \right) }
\leq  1_{\mathcal E^t_{\theta}}\mathbb E\left[e^{t \lambda^\top X_{i,t}} | W_1,...,W_t\right] \leq
1_{\mathcal E^t_{\theta}} e^{ t \overline \Upsilon_t\left(\lambda \right)}.
\end{equation}
Finally, by monotonicity of the expectation:
\begin{align*}
\mathbb P\left( \mathcal E^t_{\theta}\right)e^{ t \underline \Upsilon_t\left(\lambda \right) }
& \leq \mathbb E\left[  1_{\mathcal E^t_{\theta}}\mathbb E\left[e^{t \lambda^\top X_{i,t}} | W_1,...,W_t\right]\right]\\
  & \leq  \mathbb P\left( \mathcal E^t_{\theta}\right) e^{t  \overline \Upsilon_t\left(\lambda \right)},
\end{align*}
which combined with~\eqref{eqn-conditioning-on-matrices-in-mathcalE} implies
\begin{equation}
e^{ t\underline \Upsilon_t\left(\lambda \right)} \leq
 \mathbb E\left[ e^{t \lambda^\top X_{i,t}}| \mathcal E^t_{\theta} \right]
\leq e^{ t\overline \Upsilon_t\left(\lambda \right)}.
\end{equation}
Now, taking the logarithm and dividing by $t$,
\begin{equation*}
\underline \Upsilon_t(\lambda) \leq  \Upsilon_t\left(\lambda\right)\leq \overline \Upsilon_t(\lambda),
\end{equation*}
and noting that
\begin{align*}
\lim_{t\rightarrow +\infty} \overline \Upsilon_t(\lambda)
& =\lim_{t\rightarrow +\infty}\underline \Upsilon_t(\lambda)\\
&= (1-\theta ) M \Lambda\left(\frac{1}{M} \lambda\right)
+ \theta N \Lambda\left(\frac{1}{N} \lambda\right),
\end{align*}
the claim of Lemma~\ref{lemma-sequence-Upsilon-t-has-a-limit} follows.
\end{proof}

By the G\"{a}rtner-Ellis theorem it follows then that the sequence of measures $\nu^{\theta}_t$ satisfies the large deviations principle\footnote{We use here the variant of the G\"{a}rtner-Ellis theorem which claims the (full) LDP for the case when the domain of the limiting function is the whole space $\mathbb R^d$, as given in~\cite{DemboZeitouni93}; see also Exercise~{2.3.20} in~\cite{DemboZeitouni93} for the statement and the sketch of the proof of this result.}, with the rate function equal to the conjugate of
\begin{equation}\label{eq-def-f-theta}
f_{\theta}(\lambda):= (1-\theta) M \Lambda\left(\frac{1}{M}\lambda\right)
+\theta N \Lambda\left( \frac{1}{N}\lambda\right),
\end{equation}
for $\lambda\in \mathbb R^d$. Therefore, for every open set $E\subseteq \mathbb R^d$, there holds
\begin{equation}
\label{eqn-LDP-for-mu-t}
\liminf_{t\rightarrow +\infty}\frac{1}{t}\log \mathbb P\left(X_{i,t}\in E | \mathcal E^t_{\theta}\right) \geq
-\inf_{x\in E} \left\{\sup_{\lambda \in \mathbb R^d} \lambda^\top x - f_{\theta}(\lambda)\right\}.
\end{equation}
We next turn to computing the probability of the event $\mathcal E^t_{\theta}$.
\begin{lemma}
\label{lemma-probability-of-mathcalE}
For any $\theta \in [0,1]$, for all $t$ sufficiently large:
\begin{equation}
\label{eqn-probability-of-mathcalE}
\frac{1}{4}  p_{\mathcal H}^{ t-\lceil \theta t \rceil} \leq
\mathbb P\left( \mathcal E^t_{\theta} \right)\leq
 p_{\mathcal H}^{ t-\lceil \theta t \rceil}.
\end{equation}
\end{lemma}

\begin{proof}
By the disjoint blocks theorem~\cite{Karr} applied to the matrices in $\mathcal T_{\theta}$ and its complement $\left\{1,...,t\right\}\setminus \mathcal T_{\theta}$, we obtain
 \begin{align}
 \label{eq-two-factors}
& \mathbb P\left( \mathcal E^t_{\theta} \right) = \mathbb P\left( \left\|\Phi( \lceil\theta t\rceil,\lceil\theta t\rceil-  o_t) -J_N   \right\|\leq\frac{1}{t}  \right) \times \nonumber\\
& \mathbb P\left( G_s \in \mathcal H,\mathrm{\;for\;} s\in \mathcal T_{\theta},\;
   \left\|[\Phi(t,t- o_t)]_C -J_M   \right\|\leq\frac{1}{t}\right).
 \end{align}
We show using~\eqref{eq-rate-of-consensus-epsilon} that the first term in the right-hand side of the preceding equality goes to $1$ as $t\rightarrow +\infty$. Fix an arbitrary $\epsilon \in (0,1)$. Then, for all $t$ sufficiently large,
\begin{align}
\label{eq-second-term}
& \mathbb P\left( \left\|\Phi( \lceil\theta t\rceil,\lceil\theta t\rceil-  o_t) -J_N   \right\|\leq\frac{1}{t}  \right)\nonumber\\
&\geq 1- K_{\epsilon}e^{-t(\mathcal J-\epsilon)}\geq 1/2.
\end{align}
Clearly, being a probability, this term is also smaller than $1$ (for all $t$). Consider now the second factor in the right-hand side of~\eqref{eq-two-factors}. Conditioning on the event $\left\{G_s \in \mathcal H,\mathrm{\;for\;} s\in \mathcal T_{\theta}\right\}$, and using the fact that the probability of this event equals $p_{\mathcal H}^{t-\lceil\theta t\rceil}$ (note that the latter holds by the independence of weight matrices, Assumption~\ref{ass-network-and-observation-model}.\ref{ass-W-t}), we obtain
\begin{align}
\label{eq-first-term}
& \mathbb P\left( G_s \in \mathcal H,\mathrm{\;for\;} s\in \mathcal T_{\theta},\; \left\|[\Phi(t,t- o_t)]_C -J_M   \right\|\leq\frac{1}{t}\right) = \nonumber \\
& \mathbb P\left(\left\|[\Phi(t,t- c_t)]_C -J_M   \right\|\leq\frac{1}{t} | G_s \in \mathcal H,\mathrm{\;for\;} s\in \mathcal T_{\theta} \right)\nonumber p_{\mathcal H}^{t-\lceil\theta t\rceil}.
\end{align}
Similarly as in~\eqref{eq-second-term}, it can be shown that the conditional probability term in~\eqref{eq-first-term}, for all $t$ sufficiently large, greater than $1/2$. On the other hand, it is obviously smaller than $1$ for all $t$. Summarizing the preceding findings, the claim of the lemma follows.
\end{proof}

To bring the two key arguments together -- Lemma~\ref{lemma-probability-of-mathcalE} and the lower bound~\eqref{eqn-LDP-for-mu-t}, we start from the following simple bound
\begin{align}
\mathbb P\left(X_{i,t}\in E\right)& \geq \mathbb P\left( \left\{ X_{i,t}\in E\right\} \cap \mathcal E^t_{\theta}\right)\nonumber\\
&= \nu^{\theta}_t(E) \mathbb P\left(\mathcal E^t_{\theta}\right).
\end{align}
From superadditivity of the $\liminf$, followed by an application of~\eqref{eqn-LDP-for-mu-t} and~\eqref{eqn-probability-of-mathcalE}, we obtain
\begin{align*}
& \liminf_{t\rightarrow +\infty} \frac{1}{t}\log \mathbb P\left(X_{i,t}\in E\right)\\
& \geq  \liminf_{t\rightarrow +\infty} \frac{1}{t} \log  \nu^{\theta}_t(E) + \lim_{t\rightarrow +\infty} \frac{1}{t}\log \mathbb P\left( \mathcal E^t_{\theta}\right)\\
& \geq -\inf_{x\in E} \left\{\sup_{\lambda \in \mathbb R^d} \lambda^\top x - f_{\theta}(\lambda)\right\}\,-\, (1-\theta)|\log p_{\mathcal H}|.
\end{align*}
The preceding inequality holds for each $\theta$ in $[0,1]$. Optimizing over all such values yields:
\begin{align*}
& \liminf_{t\rightarrow +\infty}\frac{1}{t}\log \mathbb P\left(X_{i,t}\in E\right)\geq\\
& -\inf_{\theta \in [0,1]} \left\{ \inf_{x\in E} \sup_{\lambda \in \mathbb R^d}
\left\{\lambda^\top x - f_{\theta}(\lambda) \right\} + (1-\theta)|\log p_{\mathcal H}|\right\}\\
& = - \inf_{x\in E} \inf_{\theta \in [0,1]} \left\{\sup_{\lambda \in \mathbb R^d}
\left\{\lambda^\top x - f_{\theta}(\lambda) \right\} + (1-\theta)|\log p_{\mathcal H}|\right\}.
\end{align*}
Now, fix $x\in E$ and consider the function
\begin{align}\label{eq-def-g}
g(\theta,\lambda)& :=\lambda^\top x - \,(1-\theta) \left(M \Lambda\left( \frac{1}{M}\lambda\right) - |\log p_{\mathcal H}|\right)\nonumber\\
& - \theta N\Lambda\left( \frac{1}{N}\lambda\right).
\end{align} As an affine function of $\theta$, $g$ is convex in $\theta$. Further, by convexity of $\Lambda$, $g$ is concave in $\lambda$, for any $\theta\in [0,1]$. Finally, sets $[0,1]$ and $\mathbb R^d$ are convex and set $[0,1]$ is compact. Thus, conditions for applying the Minimax theorem~\cite{Sion-minimax} are fulfilled and we obtain:
\begin{align*}
& \inf_{\theta \in [0,1]} \sup_{\lambda \in \mathbb R^d}
\lambda^\top x - (1-\theta) \left(M \Lambda\left( 1/M\lambda\right) - |\log p_{\mathcal H}|\right)\\
& \;\;\;\;\;\;\;\;\;\;\;\;\;\;\;\;\; - \theta N\Lambda\left( 1/N\lambda\right) =\\
&\sup_{\lambda \in \mathbb R^d} \inf_{\theta \in [0,1]}
\lambda^\top x - (1-\theta) \left(M \Lambda\left( 1/M\lambda\right) - |\log p_{\mathcal H}|\right)\\
& \;\;\;\;\;\;\;\;\;\;\;\;\;\;\;\;\; - \theta N\Lambda\left( 1/N\lambda\right)\\
& =  \sup_{\lambda \in \mathbb R^d} \lambda^\top x -
\max\left\{ M \Lambda\left( 1/M\lambda\right)-|\log p_{\mathcal H}|, \Lambda\left( 1/N\lambda\right) \right\}.
\end{align*}
Similarly as in the proof of the upper bound, using the conjugacy rules from Lemma~\ref{lemma-conjugate-calculus},
\begin{align*}
& \sup_{\lambda \in \mathbb R^d} \lambda^\top x -
\min\left\{ M \Lambda\left( \frac{1}{M}\lambda\right)-|\log p_{\mathcal H}|, \Lambda\left( \frac{1}{N}\lambda\right) \right\} \\
& = \co \inf\left( N I, M I+ |\log p_{\mathcal H}|\right)(x),
\end{align*}
which finally yields,
\begin{align*}
& \liminf_{t\rightarrow +\infty}\frac{1}{t}\log \mathbb P\left(X_{i,t}\in E\right)\\
& \geq - \inf_{x\in E} \co \inf\left\{ N I, M I+ |\log p_{\mathcal H}|\right\}(x).
\end{align*}
This completes the proof of the lower bound and the proof of Theorem~\ref{theorem-nice-tight-bounds}.
%
%

\section{Conclusion}
\label{sec-Concl}
We studied large deviations inaccuracy rates for consensus+innovations based distributed inference for generic random networks. We assume vector measurements with possibly non-i.i.d. entries. Our goal was to find bounds or exact rate function for each node in the network, accounting for the specificities of the node's interactions. For each node, we found a node-specific family of lower bounds, induced by the family of network subgraphs in which the node participates. Specifically, each bound in the family is given as the convex envelope of the centralized rate function and the effective rate function corresponding to a given subgraph, and lifted by the probability that this subgraph remains isolated from the remainder of the network. The upper bound is defined as the convex envelope of the centralized rate function and the rate function corresponding to an isolated node, lifted by the rate of consensus. We show that, for certain cases such as pendant nodes and $d$-cyclic graphs, the two bounds match, hence proving the large deviations principle for these classes of random networks. We illustrate the results with an application to social learning, providing also the first proof of the large deviations principle for social learning beliefs with random network models.   

\begin{appendices}
%

\section{Proof of~\eqref{eq-rate-function-bounds}}
\label{app:rate-function-bounds}
Fix $i\in V$ and suppose that the inequalities in~\eqref{eq-rate-UB} and~\eqref{eq-rate-LB} hold for any set $D$. Suppose also that the sequence of node $i$'s states, $X_{i,t}$, satisfies the LDP with rate function $I_i$.

We prove~\eqref{eq-rate-function-bounds} by contradiction. Consider first the right hand side of~\eqref{eq-rate-function-bounds} and suppose, for the sake of contradiction, that there exists a point $x_0$ such that $I_i(x_0)> \overline I_i(x_0)$. Let $\epsilon=I_i(x_0)- \overline I_i(x_0)$ and introduce $S=\left\{x\in \mathbb R^d: I_i(x)> \overline I_i(x_0) + \epsilon/2\right\}$. By the lower semi-continuity of $I_i$, $S$ is open. Also, $x_0\in S$. Thus, for $\delta>0$ sufficiently small, the closed ball $\overline B_{x_0}(\delta)$ entirely belongs to $S$. Combining the LDP upper bound~\eqref{eqn-LDP-UB} for $D= \overline B_{x_0}(\delta)$, with the bound~\eqref{eq-rate-UB} for $D= B_{x_0}(\delta)$, we obtain:
\begin{align}
\label{eq-contradiction}
& -\inf_{x\in B_{x_0}(\delta)} \overline I_i(x)\leq \liminf_{t\rightarrow +\infty} \,\frac{1}{t}\log \mathbb P\left(X_{i,t}\in B_{x_0}(\delta)\right) \\
& \leq \limsup_{t\rightarrow +\infty} \,\frac{1}{t}\log \mathbb P\left(X_{i,t}\in \overline B_{x_0}(\delta)\right)\leq -\inf_{x\in \overline B_{x_0}(\delta)} I_i(x).
\end{align}
Since $\inf_{x\in B_{x_0}(\delta)} \overline I_i(x) \leq \overline I_i(x_0)$, we have that the left hand side in~\eqref{eq-contradiction} is greater than $- \overline I_i(x_0)$. On the other hand, for any $x\in \overline B_{x_0}(\delta)$, $I_i(x)>\overline I_i(x_0) + \epsilon/2$, implying $\inf_{x\in \overline B_{x_0}(\delta)} I_i(x) \geq \overline I_i(x_0) + \epsilon/2$. This finally yields contradiction since the right hand side in~\eqref{eq-contradiction} cannot be smaller than $-\overline I_i(x_0)$.

\section{Proof of Lemma~\ref{lemma-Gaussian-design}}
\label{app:compute-I-star}
We start by noting that $\epi \inf \left\{NI, I +\mathcal J\right\}= S_1 \cup S_2$, where $S_1$ and $S_2$ are the epigraphs of $NI$ and $I+\mathcal J$, $S_1= \epi (NI)$ and $S_2= \epi (I+\mathcal J)$. To prove Lemma~\ref{lemma-Gaussian-design}, we need to show that $\epi F = \co \left\{ S_1\cup S_2 \right\}$, where $F$ is the function defined in the right hand side of eq.~\eqref{eqn-characterization-of-I-star}.
To do this it suffices to show that: 1) $\epi F$ is a convex set, and 2) $\epi F \subseteq \mathrm{co} \left(S_1\cup S_2\right)$. We first prove 1). It suffices to show that $F$ is convex, which we do using generalized second order characterizations of convex functions, e.g.~\cite{Cvx-generalized-conds}. Note that $F$ is continuous and that $\mathcal D_F=\mathbb R^d$. For each $x$ and $d$, let $F^\prime_{+}(x,d)$ and $F^{\prime\prime}_{+}(x,d)$ denote, respectively, the upper directional derivatives of the first and the second order at the point $x$ and in the direction $d$,
\begin{align}
F^\prime_{+}(x;d) & = \limsup_{\epsilon\downarrow 0}\frac{F(x+\epsilon d)-F(x)}{\epsilon}\\
F^{\prime \prime}_{+}(x;d) & = \limsup_{\epsilon\downarrow 0}\frac{F(x+\epsilon d)-F(x) - F^{\prime}_{+}(x;d)}{2\,\epsilon^2}.
\end{align}
We will show that $F$ is in fact differentiable. Then, by Theorem 2.1. part~(i) from~\cite{Cvx-generalized-conds}, proving convexity of $F$ would reduce to proving that $F^{\prime \prime}_{+}(x;d)\geq 0$ for any $x$ and $d$. Note that $I$ and $H$ are differentiable, with their respective gradients given by $\nabla I(x)= S^{-1}(x-m)$ and $\nabla H(x)= S^{-1}(x-m)/\sqrt{(x-m)^\top S^{-1}{x-m}}$. Thus, $F$ is differentiable in each of the three open sets (note that $I$ is continuous and differentiable): $\left\{x: I(x)< c_1\right\}$, $\left\{x: c_1<I(x)< N c_1\right\}$, and $\left\{x: N I(x)> c_1\right\}$. It remains to show that $F$ is differentiable for those $x$ such that $ I(x)=c_1$ and $I(x)=Nc_1$. Fix first $x$ such that $I(x) = c_1$.
It is easy to see that, for any $d$ such that $d^\top S^{-1}(x-m)\geq 0$, $I(x+\epsilon d) > I(x)$ for all $\epsilon>0$. Also, for any $d$ such that $d^\top S^{-1}(x-m)< 0$, $I(x+\epsilon d) < I(x)$ for all sufficiently small $\epsilon>0$. Thus, if $d^\top S^{-1}(x-m)\geq 0$, $F(x+\epsilon d)= N \sqrt{2 c_1}H(x+\epsilon d)- N c_1$, for all $\epsilon$ sufficiently small, and hence $F^\prime_{+}(x;d)= N \sqrt{2 c_1} d^\top \nabla H(x)$. Using now the fact that $I(x)=c_1$, we obtain that $F^\prime_{+}(x;d) = N d^\top S^{-1}(x-m)$. Consider now the case when $d$ is such that $d^\top S^{-1}(x-m)\leq 0$. Then, by the discussion above we have that for all $\epsilon $, $F(x+\epsilon d)= N I(x+\epsilon d)$. Hence, $F^\prime_{+}(x;d)= N d^\top \nabla I(x)=N d^\top S^{-1}(x-m)$. Since for any $x$ s.t. $I(x)=c_1$ and for any $d$ we have that $F^\prime_{+}(x;d) = N d^\top \nabla I(x)$, we conclude that $F$ is differentiable at any such $x$. We can in analogous manner prove differentiability of $F$ at any $x$ s.t. $I(x)=N c_1$. Hence, we conclude that $F$ is differentiable.

We now turn to proving that $F^{\prime\prime}_{+}(x;d)\geq 0$ for any $x$ and $d$. Note that $\nabla^2 I(x)=S^{-1}\succeq 0$ and
\begin{align*}
& \nabla^2 H(x) =\\
&  N \frac{\sqrt{2 c_1}} {\sqrt{2I(x)}}\left(S^{-1} - \frac{1}{2I(x)} S^{-1} (x-m)(x-m)^\top S^{-1}\right),
\end{align*}
for any $x$. To see that $\nabla^2 H(x)\succeq 0$, it suffices to observe that it can be rewritten as $\nabla^2 H(x) = N \sqrt{2 c_1} /\sqrt{2 I(x)}S^{-1/2} (I- qq^\top/(\|q\|^2)) S^{-1/2}$, for $q=S^{-1/2}(x-m)$. Since the matrix inside the brackets is positive semidefinite, positive semidefiniteness of $\nabla^2 H(x)$ follows. Therefore, for any $x$ in the interior of the three sets in~\eqref{eqn-characterization-of-I-star}, we have that $F^{\prime \prime}_{+}(x;d)\geq 0$. Consider now the case when $x$ satisfies $I(x)=c_1$. Following the same steps as in the preceding paragraph, we obtain that for any $d$ s.t. $d^\top S^{-1}(x-m)\geq 0$, $F^{\prime \prime}_{+}(x;d)= N^2 2c_1 d^\top \nabla^2 H(x) d \geq 0$ and for $d$ s.t. $d^\top S^{-1}(x-m)\leq 0$, $F^{\prime \prime}_{+}(x;d)= N d^\top \nabla^2 I(x) d \geq 0$. To complete the proof of 1), it only remains to consider those $x$ that satisfy $I(x)=Nc_1$. Analogously to the preceding case, we get that for $d$ s.t. $d^\top S^{-1}(x-m)\geq 0$, $F^{\prime \prime}_{+}(x;d)= d^\top \nabla^2 I(x) d \geq 0$ and for $d$ s.t. $d^\top S^{-1}(x-m)\leq 0$, $F^{\prime \prime}_{+}(x;d)= n^2 2 c_1 d^\top \nabla^2 H(x) d \geq 0$. Hence, since $F$ is differentiable and $F^{\prime \prime}_{+}(x;d)\geq 0$ for any $x$ and $d$, we conclude that $F$ is convex.

To prove Lemma~\ref{lemma-Gaussian-design}, it remains to prove part 2). For each unit norm $v\in \mathbb R^d$, $\|v\|=1$, let $\phi_v:\mathbb R^d\mapsto \mathbb R^d$ denote the projection of $F$ along the direction $v$, started at point $m$: $\phi_v(\rho):=F(m+\rho v)$, $\rho\in \mathbb R$. Then, $\epi F=\cup_{v\in \mathbb R^d, \|v\|=1} \mathrm{epi} \phi_v$. For each fixed $v$, let $[S_l]_v$ denote the projection of $S_l$ along the line $m+\rho v$, $[S_l]_v= S_l\cap \left\{m+\rho v: \rho\in \mathbb R\right\}$, $l=1,2$. Note that  $[S_1]_v= \left\{ (t,m+\rho v): \,t\geq N \rho^2 v^\top S^{-1} v/2,\,\rho\in \mathbb R\right\}$, $[S_2]_v= \left\{ (t,m+\rho v): \,t\geq \rho^2 v^\top S^{-1} v/2 +\mathcal J,\,\rho\in \mathbb R\right\}$. Then, it is easy to see that, for each unit norm $v$, $\mathrm{epi}\phi_v=\mathrm{co} \left([S_1]_v\cup [S_2]_v\right)$. Finally, since $\mathrm{co} \left([S_1]_v\cup [S_2]_v\right)\subseteq \mathrm{co} \left(S_1\cup S_2\right)$, the claim in 2) follows. This completes the proof of Lemma~\ref{lemma-Gaussian-design}.

\section{Proof of Lemma~\ref{theorem-LD-rates-for-beliefs-in-SL}}
\label{app:lemma-rates-for-SL-beliefs}
Fix an arbitrary node $i\in V$. For each $m=1,...,M-1$, $X_{i,t}^m = \frac{1}{t} \log \frac{b_{i,t}^m}{b_{i,t}^M}$, hence  
\begin{equation}
\label{eq-X-b-m-relation}
\frac{1}{t}\log b_{i,t}^m = X_{i,t}^m + \frac{1}{t}\log b_{i,t}^M.     
\end{equation}
Further, the (private) beliefs by construction sum up to one: $\sum_{m=1}^M b_{i,t}^m = 1$. Dividing both sides by $b_{i,t}^M$ and exploiting the functional relation between $b_{i,t}^m$ and $X_{i,t}^m$, we obtain
\begin{equation}
\label{eq-X-b-relation}
\sum_{m=1}^{M-1}e^{t X_{i,t}^m}+1 = \frac{1}{b_{i,t}^M}.     
\end{equation}
It follows that:
\begin{equation}
\label{eq-X-b-relation-2}
\frac{1}{M} e^{-t \max_{m=1,...,M} X_{i,t}^m} \leq b_{i,t}^M \leq e^{-t \max_{m=1,...,M} X_{i,t}^m},     
\end{equation}
where $X_{i,t}^M\equiv 0$. From~\eqref{eq-X-b-m-relation} and~\eqref{eq-X-b-relation-2} we obtain
\begin{equation}
\label{eq-g-m-relation}
g_m(X_{i,t}) -\frac{1}{t}\log M \leq \frac{1}{t}\log b_{i,t}^m \leq g_m(X_{i,t}).    
\end{equation}
Consider now an arbitrary one-sided closed interval $F$ on $\mathbb R$. Suppose that $F=[a,+\infty)$ (other intervals in $\mathbb R$ can be treated analogously). Fix $\epsilon>0$. From~\eqref{eq-g-m-relation}, for all $t\geq t_0 = \log M/\epsilon$ there holds:
\begin{equation}
\label{eq-g-m-relation-epsilon}
g_m(X_{i,t}) -\epsilon \leq \frac{1}{t}\log b_{i,t}^m \leq g_m(X_{i,t}),    
\end{equation}
and thus, for all $t\geq t_0$
\begin{equation}
\label{eq-g-m-relation-epsilon-prob}
\mathbb P(\frac{1}{t}\log b_{i,t}^m \geq a+\epsilon) \leq \mathbb P(g_m(X_{i,t}) \geq a) = P(X_{i,t} \in g_m^{-1}([a, +\infty)).     
\end{equation}
Taking the $\limsup$ over $t\rightarrow +\infty$,
\begin{equation}
\label{eq-g-m-relation-epsilon-prob-2}
\limsup_{t\rightarrow+\infty} \frac{1}{t}\log \mathbb P(\frac{1}{t}\log b_{i,t}^m \geq a+\epsilon) \leq \limsup_{t\rightarrow+\infty} \frac{1}{t}\log P(X_{i,t} \in g_m^{-1}([a, +\infty)).     
\end{equation}
The above inequality holds for all $\epsilon>0$. Taking the supremum over $\epsilon>0$ on the left hand side yields:
\begin{equation}
\label{eq-g-m-relation-epsilon-prob-3}
\limsup_{t\rightarrow+\infty} \frac{1}{t}\log \mathbb P(\frac{1}{t}\log b_{i,t}^m \geq a) \leq \limsup_{t\rightarrow+\infty} \frac{1}{t}\log P(X_{i,t} \in g_m^{-1}([a, +\infty)).     
\end{equation}
Applying now the upper bound in~\ref{eq-UB-SL}, the upper bound in~\eqref{eq-UB-SL-beliefs} follows. The proof of the lower bound~\eqref{eq-UB-SL-beliefs} is analogous.   

\section{Proof of Lemma~\ref{lemma-UB-skeleton-d>1}}
\label{app:UB-skeleton}
Suppose that $X_t\in \mathbb R^d$ is a sequence of random variables for which~\eqref{eq-if-UB} holds for some function $f$.  Fix a compact set $F\subseteq \mathbb R^d$. For each $\delta>0$, introduce the function $f^{\star,\delta}:\mathbb R^d\mapsto \mathbb R$ obtained by truncating $f^\star$ to $1/\delta$:
\begin{equation}
f^{\star,\delta}(x)= \inf\left\{\frac{1}{\delta}, f^\star(x)-\delta\right\},\mbox{\;for\;}x\in \mathbb R^d.
\end{equation}
The family of functions $f^{\star,\delta}$, $\delta>0,$ satisfies that, for any set $D$,
\begin{equation}
\label{eq-limit-inf-I-delta}
\lim_{\delta\rightarrow 0} \inf_{x\in D}f^{\star,\delta}(x)= \inf_{x\in D} f^\star(x).
\end{equation}
To show this, let $\xi:=\inf_{x\in D}f^\star(x)$ and suppose first that $\xi=+\infty$, i.e., $f^\star$ at all points $x\in D$ takes the value $+\infty$. Then, for any $\delta>0$, $f^{\star,\delta} = 1/\delta$ for all $x\in D$, and therefore, for any $\delta>0$, $\inf_{x\in D} f^{\star,\delta}(x)=1/\delta$. Computing the limit $\lim_{\delta\rightarrow 0}1/\delta=+\infty$, identity~\eqref{eq-limit-inf-I-delta} follows. We next consider the case  $\xi\in \mathbb R$. For arbitrary fixed $\delta>0$, the quantity under the limit in the left hand side of~\eqref{eq-limit-inf-I-delta} equals:
\begin{align}
\label{eq-auxil-1}
\inf_{x\in D}f^{\star,\delta}(x)&= \inf_{x\in D} \inf\left\{ f^\star(x)-\delta, \frac{1}{\delta}\right\}\nonumber\\
&=  \inf\left\{ \inf_{x\in D} \left( f^\star(x)-\delta\right), \frac{1}{\delta}\right\}.
\end{align}
The first argument of the infimum~\eqref{eq-auxil-1} equals $\xi-\delta$ and it is finite by our assumption. Hence, for all $\delta$ sufficiently small, the infimum~\eqref{eq-auxil-1} equals $\xi-\delta$, which after taking the limit $\delta\rightarrow 0$ yields the claim.
The case $\xi=-\infty$ can be proven equivalently.

Having~\eqref{eq-limit-inf-I-delta}, it easy to see that~\eqref{eq-then-UB} follows if we show that the following inequality holds for any given $\delta$:
\begin{equation}
\label{eq-then-UB-it-suffices-delta}
\limsup_{t\rightarrow +\infty} \,\frac{1}{t}\,\log \mathbb P\left(X_t\in F\right) \leq 2\delta -\inf_{x\in F} f^{\star,\delta}(x).
\end{equation}
Thus, in what follows we focus on proving~\eqref{eq-then-UB-it-suffices-delta}. To this end, fix $\delta>0$. For any point $y\in F$ there exists a point $\lambda_y$ (which depends on $\delta$) such that
\begin{equation}
\label{eq-approach-with-lambda-y}
\lambda_y^\top y - \Lambda^\star(\lambda_y) \geq f^{\star,\delta}(y).
\end{equation}
Existence of such a point follows directly from the definitions of $f^\star$ and $f^{\star,\delta}$. First, since for any fixed point $y$ $f^\star(y)$ is computed as the supremum of functions $\lambda \mapsto h_y(\lambda):=\lambda^\top y - f(\lambda)$, it follows that the value $f^\star(y)$ can be approached arbitrarily close with $h_y(\lambda)$. Second, since $f^{\star,\delta}(y)$ is the infimum of $f^\star(y)-\delta$ and $1/\delta$, it must satisfy $f^\star(y)-\delta,\,1/\delta \geq f^{\star,\delta}(y)$. For example, if, for some $y$, $f^\star(y)$ is finite, then there must exist a point $\lambda$ such that $h_y(\lambda)\geq f^\star(y)-\delta$, and since the latter is greater than $f^{\star,\delta}(y)$,~\eqref{eq-approach-with-lambda-y} follows.

Note now that~\eqref{eq-if-UB} implies that, for any measurable set $D$, there exists $t_0=t_0(\delta, D)$ such that
\begin{equation}
\label{eq-if-UB-finite-time}
\frac{1}{t}\,\log \mathbb P\left(X_t\in D\right) \leq \delta + f(\lambda) - \inf_{x\in D} \lambda^\top x,
\end{equation}
for all $t\geq t_0$. For any $y\in F$, let $r_y:=\delta/\|\lambda_y\|$. Taking $D=\overline B_y(r_y)$ and $\lambda=\lambda_y$ in~\eqref{eq-if-UB-finite-time} yields for any $t\geq t_0(\delta,y)$
\begin{align}
\label{eq-if-UB-finite-time-ball-y}
\frac{1}{t}\,\log \mathbb P\left(X_t\in \overline B_y(r_y)\right) & \leq \delta + f(\lambda_y) - \inf_{ \|x-y\|\leq r_y} \lambda_y^\top x\\
& \leq \delta + \Lambda^\star(\lambda_y) - \lambda_y^\top y - \inf_{\|x\|\leq r_y} \lambda_y^\top x\\
& \leq 2 \delta - f^{\star,\delta}(y),
\end{align}
where the last inequality follows from~\eqref{eq-approach-with-lambda-y} and the definition of $r_y$. Next, from the family of closed balls $\left\{\overline B_y(r_y)\,:\, y\in F\right\}$, a finite cover of $F$, $\left\{\overline B_{y_k}(r_{y_k})\,:\,k=1,...,K\right\}$, is extracted, where, we note, $K=K(F, \delta)$. Then, by the union bound,
\begin{align*}
\frac{1}{t}\log \,\mathbb P& \left(X_t\in F\right) \leq \frac{1}{t}\log \left(\sum_{k=1}^K \mathbb P\left(X_t\in \overline B_{y_k}(r_{y_k})\right)\right)\nonumber\\
&\leq  \frac{1}{t}\log K + \frac{1}{t}\log \max_{k=1,...,K} \mathbb P\left(X_t\in \overline B_{y_k}(r_{y_k})\right)\nonumber\\
&\leq  \frac{1}{t}\log K + \max_{k=1,...,K} \frac{1}{t}\log \mathbb P\left(X_t\in \overline B_{y_k}(r_{y_k})\right).\nonumber\\
\end{align*}
Combining the preceding inequality with~\eqref{eq-if-UB-finite-time-ball-y} applied for every $k=1,...,K$, we have that for every $t\geq \max_{k=1,...K}t_0(\delta,y_k)$
\begin{align}
\frac{1}{t}\log \,\mathbb P\left(X_t\in F\right) & \leq  \frac{1}{t}\log K + \max_{k=1,...,K} 2 \delta - f^{\star,\delta}(y)\nonumber\\
& \leq  \frac{1}{t}\log K  + 2\delta - \inf_{y\in F} f^{\star,\delta}(y).
\end{align}
Taking the limit $t\rightarrow +\infty$, and noting that $K$ is finite,~\eqref{eq-then-UB-it-suffices-delta} follows. The last part of the claim, i.e.,~\eqref{eq-then-UB-it-suffices-delta} for closed sets follows from~\eqref{eq-then-UB-it-suffices-delta} for compact sets, that we have just proved, and Lemma 1.2.18 in~\cite{DemboZeitouni93}.

\end{appendices}

\bibliographystyle{IEEEtran}
\bibliography{IEEEabrv,Bibliography}

\begin{thebibliography}{10}
\providecommand{\url}[1]{#1}
\csname url@samestyle\endcsname
\providecommand{\newblock}{\relax}
\providecommand{\bibinfo}[2]{#2}
\providecommand{\BIBentrySTDinterwordspacing}{\spaceskip=0pt\relax}
\providecommand{\BIBentryALTinterwordstretchfactor}{4}
\providecommand{\BIBentryALTinterwordspacing}{\spaceskip=\fontdimen2\font plus
\BIBentryALTinterwordstretchfactor\fontdimen3\font minus
  \fontdimen4\font\relax}
\providecommand{\BIBforeignlanguage}[2]{{%
\expandafter\ifx\csname l@#1\endcsname\relax
\typeout{** WARNING: IEEEtran.bst: No hyphenation pattern has been}%
\typeout{** loaded for the language `#1'. Using the pattern for}%
\typeout{** the default language instead.}%
\else
\language=\csname l@#1\endcsname
\fi
#2}}
\providecommand{\BIBdecl}{\relax}
\BIBdecl

\bibitem{BajovicThesis13}
D.~Bajovi\'{c}, ``Large deviations rates for distributed inference,'' Ph.D.
  dissertation, Carnegie Mellon University, 2013.

\bibitem{Chernoff52}
H.~Chernoff, ``A measure of the asymptotic efficiency of tests of a hypothesis
  based on a sum of observations,'' \emph{The Annals of Mathematical
  Statistics}, vol.~23, no.~4, pp. 493--507, Dec. 1952.

\bibitem{Anandkumar07}
A.~Anandkumar and L.~Tong, ``Type-based random access for distributed detection
  over multiaccess fading channels,'' \emph{IEEE Transactions on Signal
  Processing}, vol.~55, no.~10, pp. 5032--5043, 2007.

\bibitem{Bajovic11}
D.~Bajovi\'{c}, B.~Sinopoli, and J.~Xavier, ``Sensor selection for event
  detection in wireless sensor networks,'' \emph{IEEE Transactions on Signal
  Processing}, vol.~59, no.~10, pp. 4938--4953, 2011.

\bibitem{Tay15}
W.~P. Tay, ``Whose opinion to follow in multihypothesis social learning? {A}
  large deviations perspective,'' \emph{IEEE Journal of Selected Topics in
  Signal Processing}, vol.~9, no.~2, pp. 344--359, 2015.

\bibitem{Ping22}
\BIBentryALTinterwordspacing
P.~Hu, V.~Bordignon, S.~Vlaski, and A.~H. Sayed, ``Optimal aggregation
  strategies for social learning over graphs,'' 2022. [Online]. Available:
  \url{https://arxiv.org/abs/2203.07065}
\BIBentrySTDinterwordspacing

\bibitem{Bucklew90}
J.~A. Bucklew, \emph{Large Deviations Techniques in Decision, Simulation and
  Estimation}.\hskip 1em plus 0.5em minus 0.4em\relax New York: Wiley, 1990.

\bibitem{SchwartzWeiss95}
A.~Shwartz and A.~Weiss, \emph{Large Deviations for Performance Analysis:
  Queues, Communications, and Computing}.\hskip 1em plus 0.5em minus
  0.4em\relax New York: Chapman and Hall, 1995.

\bibitem{Touchette2009LDStatMechs}
H.~Touchette, ``The large deviation approach to statistical mechanics,''
  \emph{Physics Reports}, vol. 478, no.~1, pp. 1 -- 69, 2009.

\bibitem{Cover91}
T.~M. Cover and J.~A. Thomas, \emph{Elements of {I}nformation {T}heory}.\hskip
  1em plus 0.5em minus 0.4em\relax New York: John Wiley and Sons, 1991.

\bibitem{Arcones06LDM-estimators}
M.~Arcones, ``Large deviations for {M}-estimators,'' \emph{Annals of the
  Institute of Statistical Mathematics}, vol.~58, no.~1, pp. 21--52, 2006.

\bibitem{Bahadur60}
\BIBentryALTinterwordspacing
R.~R. Bahadur, ``On the asymptotic efficiency of tests and estimates,''
  \emph{Sankhya: The Indian Journal of Statistics, 1933-1960}, vol.~22, no.
  3/4, pp. 229--252, 1960. [Online]. Available:
  \url{http://www.jstor.org/stable/25048458}
\BIBentrySTDinterwordspacing

\bibitem{GaussianDD}
D.~Bajovi\'{c}, D.~Jakoveti\'{c}, J.~Xavier, B.~Sinopoli, and J.~M.~F. Moura,
  ``Distributed detection via {G}aussian running consensus: Large deviations
  asymptotic analysis,'' \emph{IEEE Transactions on Signal Processing},
  vol.~59, no.~9, pp. 4381--4396, Sep. 2011.

\bibitem{Non-Gaussian-DD}
D.~Bajovi\'{c}, D.~Jakoveti\'{c}, J.~M.~F. Moura, J.~Xavier, and B.~Sinopoli,
  ``Large deviations performance of consensus+innovations distributed detection
  with non-{G}aussian observations,'' \emph{IEEE Transactions on Signal
  Processing}, vol.~60, no.~11, pp. 5987--6002, Nov. 2012.

\bibitem{DDNoisy12}
D.~{Jakoveti\'c}, J.~M.~F. {Moura}, and J.~{Xavier}, ``Distributed detection
  over noisy networks: Large deviations analysis,'' \emph{IEEE Transactions on
  Signal Processing}, vol.~60, no.~8, pp. 4306--4320, 2012.

\bibitem{DI-Directed-Networks16}
D.~Bajovi\'{c}, J.~M.~F. Moura, J.~Xavier, and B.~Sinopoli, ``Distributed
  inference over directed networks: Performance limits and optimal design,''
  \emph{IEEE Transactions on Signal Processing}, vol.~64, no.~13, pp.
  3308--3323, July 2016.

\bibitem{MBMS16InfTheory}
V.~Matta, P.~Braca, S.~Marano, and A.~H. Sayed, ``Diffusion-based adaptive
  distributed detection: Steady-state performance in the slow adaptation
  regime,'' \emph{IEEE Trans. Information Theory}, vol.~62, no.~8, pp.
  4710--4732, August 2016.

\bibitem{MBMS16Refined}
------, ``Distributed detection over adaptive networks: Refined asymptotics and
  the role of connectivity,'' \emph{IEEE Trans. Signal and Information
  Processing over Networks}, vol.~2, no.~4, pp. 442--460, Dec 2016.

\bibitem{MaranoSayed19OneBit}
S.~Marano and A.~H. Sayed, ``Detection under one-bit messaging over adaptive
  networks,'' \emph{IEEE Trans. Information Theory}, vol.~65, no.~10, pp.
  6519--6538, October 2019.

\bibitem{Jadb2012NonBayesian}
\BIBentryALTinterwordspacing
A.~Jadbabaie, P.~Molavi, A.~Sandroni, and A.~Tahbaz-Salehi, ``Non-{B}ayesian
  social learning,'' \emph{Games and Economic Behavior}, vol.~76, no.~1, pp.
  210 -- 225, 2012. [Online]. Available:
  \url{http://www.sciencedirect.com/science/article/pii/S0899825612000851}
\BIBentrySTDinterwordspacing

\bibitem{Shahrampour16DDFiniteTime}
S.~{Shahrampour}, A.~{Rakhlin}, and A.~{Jadbabaie}, ``Distributed detection:
  Finite-time analysis and impact of network topology,'' \emph{IEEE
  Transactions on Automatic Control}, vol.~61, no.~11, pp. 3256--3268, 2016.

\bibitem{Lalitha18SLandDHT}
A.~{Lalitha}, T.~{Javidi}, and A.~D. {Sarwate}, ``Social learning and
  distributed hypothesis testing,'' \emph{IEEE Transactions on Information
  Theory}, vol.~64, no.~9, pp. 6161--6179, 2018.

\bibitem{Mitra21}
A.~Mitra, J.~A. Richards, and S.~Sundaram, ``A new approach to distributed
  hypothesis testing and non-bayesian learning: Improved learning rate and
  byzantine resilience,'' \emph{IEEE Transactions on Automatic Control},
  vol.~66, no.~9, pp. 4084--4100, 2021.

\bibitem{DeGroot74}
M.~H. DeGroot, ``Reaching a consensus,'' \emph{Journal of American Statistical
  Association}, vol.~69, no. 345, pp. 118--121, 1974.

\bibitem{NedicSL-TV17}
A.~{Nedić}, A.~{Olshevsky}, and C.~A. {Uribe}, ``Fast convergence rates for
  distributed non-{B}ayesian learning,'' \emph{IEEE Transactions on Automatic
  Control}, vol.~62, no.~11, pp. 5538--5553, 2017.

\bibitem{Rate-of-consensus13}
D.~Bajovi\'{c}, J.~Xavier, J.~M.~F. Moura, and B.~Sinopoli, ``Consensus and
  products of random stochastic matrices: Exact rate for convergence in
  probability,'' \emph{IEEE Transactions on Signal Processing}, vol.~61,
  no.~10, pp. 2557--2571, May 2013.

\bibitem{Parasnis20}
\BIBentryALTinterwordspacing
R.~Parasnis, M.~Franceschetti, and B.~Touri, ``Non-{B}ayesian social learning
  on random digraphs with aperiodically varying network connectivity,'' 2020.
  [Online]. Available: \url{https://arxiv.org/abs/2010.06695}
\BIBentrySTDinterwordspacing

\bibitem{DemboZeitouni93}
A.~Dembo and O.~Zeitouni, \emph{Large {D}eviations {T}echniques and
  {A}pplications}.\hskip 1em plus 0.5em minus 0.4em\relax {B}oston, {MA}: Jones
  and Barlett, 1993.

\bibitem{Soummya-LDRiccati14}
D.~Li, S.~Kar, {J.~M.~F.~Moura}, {H.~V.~Poor}, and S.~Cui, ``Distributed
  {K}alman filtering over massive data sets: Analysis through large deviations
  of random {R}iccati equations,'' \emph{IEEE Transactions on Information
  Theory}, vol.~61, no.~3, pp. 1351--1372, March 2015.

\bibitem{Hollander}
{F.~den~Hollander}, \emph{Large {D}eviations}.\hskip 1em plus 0.5em minus
  0.4em\relax Fields {I}nstitute {M}onographs, American {M}athematical
  {S}ociety, 2000.

\bibitem{Vysotsky21}
\BIBentryALTinterwordspacing
V.~Vysotsky, ``{When is the rate function of a random vector strictly
  convex?}'' \emph{Electronic Communications in Probability}, vol.~26, no.
  none, pp. 1 -- 11, 2021. [Online]. Available:
  \url{https://doi.org/10.1214/21-ECP409}
\BIBentrySTDinterwordspacing

\bibitem{Urruty}
J.-B. Hiriart-Urruty and C.~Lemarechal, \emph{Fundamentals of Convex Analysis},
  ser. Grundlehren Text Editions.\hskip 1em plus 0.5em minus 0.4em\relax
  {B}erlin, {G}ermany: Springer-Verlag, 2004.

\bibitem{Rockafellar2015}
\BIBentryALTinterwordspacing
R.~T. Rockafellar, \emph{Convex Analysis}.\hskip 1em plus 0.5em minus
  0.4em\relax Princeton University Press, 2015. [Online]. Available:
  \url{https://doi.org/10.1515/9781400873173}
\BIBentrySTDinterwordspacing

\bibitem{Karr}
A.~F. Karr, \emph{Probability}, ser. Springer Texts in Statistics.\hskip 1em
  plus 0.5em minus 0.4em\relax {N}ew {Y}ork: Springer-Verlag, 1993.

\bibitem{Allerton11}
D.~Bajovi\'c, D.~Jakoveti\'c, J.~M.~F. Moura, J.~Xavier, , and B.~Sinopoli,
  ``Large deviations analysis of consensus+innovations detection in random
  networks,'' in \emph{Allerton'11, 49th Allerton Conference on Communication,
  Control, and Computing}, Monticello, Il, October 2011.

\bibitem{Sion-minimax}
M.~Sion, ``On general minimax theorems,'' \emph{Pacific Journal of
  Mathematics}, vol.~8, no.~1, pp. 171--176, March 1958.

\bibitem{Cvx-generalized-conds}
I.~Ginchev and V.~I. Ivanov, ``Second-order characterizations of convex and
  pseudoconvex functions,'' \emph{Journal of Applied Analysis}, vol.~9, no.~2,
  pp. 261--273, June 2010.

\end{thebibliography}
\end{document}